\documentclass[journal,twocolumn]{IEEEtran}
\usepackage{amsmath,graphicx}
\usepackage{color}
\usepackage{graphicx}
\usepackage{epstopdf}
\usepackage{amsmath}
\usepackage{amssymb}
\usepackage{mathrsfs}
\usepackage{hyperref}
\usepackage{tikz}
\usepackage{bm}
\usepackage[english]{babel}
\usepackage{cite}
\usepackage{rotfloat}
\usepackage{mathtools}
\usepackage[font=normalsize,labelfont=bf]{caption}
\usepackage{amsmath}
\usepackage{makecell}
\usepackage{algorithm,algorithmic}
\usepackage{multirow}
\usepackage{subfigure}
\usepackage{booktabs}
\usepackage{colortbl}
\usepackage{multirow}
\usepackage{hhline}
\usepackage{stfloats}
\usepackage{multicol}
\usepackage{bbm}
\usepackage{cases}
\graphicspath{ {Figures/} }
\newcommand{\T}{{\scriptscriptstyle\mathsf{T}}}
\renewcommand{\H}{{\scriptscriptstyle\mathsf{H}}}
\allowdisplaybreaks
\newsavebox{\foobox}

\definecolor{kugray5}{RGB}{224,224,224}

\usepackage[normalem]{ulem}
\newcommand\rsout{\bgroup\markoverwith
	{\textcolor{red}{\rule[0.5ex]{2pt}{0.8pt}}}\ULon}



\makeatletter
\newcommand{\ALOOP}[1]{\ALC@it\algorithmicloop\ #1%
	\begin{ALC@loop}}
	\newcommand{\ENDALOOP}{\end{ALC@loop}\ALC@it\algorithmicendloop}
\renewcommand{\algorithmicrequire}{\textbf{Input:}}

\makeatother

\usepackage{etoolbox}
\let\mybibitem\bibitem
\renewcommand{\bibitem}[1]{%
	\ifstrequal{#1}{nature}
	{\color{blue}\mybibitem{#1}}
	{\color{black}\mybibitem{#1}}%
}

\graphicspath{ {Figures/} }

\newtheorem{theorem}{\textbf{Theorem}}

\newtheorem{lemma}{\textbf{Lemma}}

\newtheorem{proof}{Proof}

\newcommand{\epr}{\hfill\(\Box\)}

\DeclareCaptionLabelSeparator{periodspace}{.\quad}

\renewcommand{\figurename}{Fig.}
\addto\captionsenglish{\renewcommand{\figurename}{Fig.}}
\newcommand\nbthis{\addtocounter{equation}{1}\tag{\theequation}}

\newcommand{\norm}[1]{\left\lVert#1\right\rVert} 
\newcommand{\abs}[1]{\left|#1\right|} 

\newcommand{\tr}[1]{\text{trace}\left(#1\right)} 
\newcommand{\trshort}[1]{\text{trace}(#1)} 

\newcommand{\diag}[1]{\mathtt{diag}\left\{#1\right\}} 
\newcommand{\diagshort}[1]{\mathtt{diag}\{#1\}} 

\newcommand{\re}[1]{\mathfrak{R}{\left(#1\right)}}
\newcommand{\im}[1]{\mathfrak{I}{\left(#1\right)}}
\allowdisplaybreaks

\newcommand{\mean}[1]{\mathbb{E} \left\{#1\right\}}
\newcommand{\meanshort}[1]{\mathbb{E} \{#1\}}

\usepackage{setspace}

\newcommand{\mR}{{\mathbf{R}}}
\newcommand{\mH}{{\mathbf{H}}} 

\newcommand{\mA}{{\mathbf{A}}}
\newcommand{\mS}{{\mathbf{S}}}
\newcommand{\mW}{{\mathbf{W}}}

\newcommand{\mI}{\textbf{\textbf{I}}}
\newcommand{\mT}{{\mathbf{T}}}

\newcommand{\mD}{{\mathbf{D}}}
\newcommand{\mX}{{\mathbf{X}}}
\newcommand{\mY}{{\mathbf{Y}}}
\newcommand{\mG}{{\mathbf{G}}}
\newcommand{\mF}{{\mathbf{F}}}
\newcommand{\mU}{{\mathbf{U}}}

\newcommand{\mZ}{{\mathbf{Z}}}

\newcommand{\mN}{{\mathbf{N}}}

\newcommand{\setC}{\mathbb{C}} 
\newcommand{\setR}{\mathbb{R}}

\newcommand{\vs}{{\mathbf{s}}}
\newcommand{\vx}{{\mathbf{x}}}

\newcommand{\vv}{{\mathbf{v}}}
\newcommand{\vn}{{\mathbf{n}}}

\newcommand{\vz}{{\mathbf{z}}} 
\newcommand{\vh}{{\mathbf{h}}}

\newcommand{\vb}{{\mathbf{b}}}
\newcommand{\vw}{{\mathbf{w}}}
\newcommand{\va}{{\mathbf{a}}}

\newcommand{\vt}{{\mathbf{t}}}

\newcommand{\vf}{\mathbf{f}}

\def\b0{{\pmb{0}}}

\newcommand{\Nr}{N_\mathtt{r}}
\newcommand{\Nt}{N_\mathtt{t}}

\newcommand{\Nrh}{N_\mathtt{rh}}
\newcommand{\Nrv}{N_\mathtt{rv}}
\newcommand{\Nth}{N_\mathtt{th}}
\newcommand{\Ntv}{N_\mathtt{tv}}

\newcommand{\ah}{\va_\mathtt{h}} 
\newcommand{\av}{\va_\mathtt{v}}

\newcommand{\sigmac}{\sigma_{\mathtt{c}}^2}
\newcommand{\sigmas}{\sigma_{\mathtt{s}}^2}

\newcommand{\noise}{\sigma^2}

\newcommand{\Ttt}{\tau_{\theta \theta}}
\newcommand{\Ttp}{\tau_{\theta \phi}}
\newcommand{\Tta}{\vt_{\theta \bar{\bm{\alpha}}}}
\newcommand{\Tpp}{\tau_{\phi \phi}}
\newcommand{\Tpa}{\vt_{\phi \bar{\bm{\alpha}}}}
\newcommand{\Taa}{\mT_{\bar{\bm{\alpha}}\bar{\bm{\alpha}}}}

\newcommand{\SEc}{\mathtt{SE}_{\mathtt{c}}} 
\newcommand{\SEth}{\mathtt{SE}_0} 
\newcommand{\SEck}{\mathtt{SE}_{\mathtt{c},k}} 
\newcommand{\whSEck}{\widehat{\mathtt{SE}}_{\mathtt{c},k}}
\newcommand{\wtSEck}{\widetilde{\mathtt{SE}}_{\mathtt{c},k}}

\newcommand{\reim}{\mathcal{X}}
\newcommand{\whreim}{\widehat{\mathcal{X}}}
\newcommand{\resym}{\mathfrak{R}}
\newcommand{\imsym}{\mathfrak{I}}
\newcommand{\slvar}[3]{\mathcal{#1}_{#2}^{#3}}

\newcommand{\Tpap}{\mathcal{T_{\phi \bar{\bm{\alpha}} \phi }}}
\newcommand{\TpapR}{\mathcal{T_{\phi \bar{\bm{\alpha}} \phi }^{\resym}}}
\newcommand{\TpapI}{\mathcal{T_{\phi \bar{\bm{\alpha}} \phi }^{\imsym}}}
\newcommand{\TpapX}{\mathcal{T_{\phi \bar{\bm{\alpha}} \phi }^{\reim}}}
\newcommand{\TpapwhX}{\mathcal{T_{\phi \bar{\bm{\alpha}} \phi }^{\whreim}}}

\newcommand{\Ldp}[2]{f_{p^2}^{\mathcal{L}} \big({#1},{#2}\big)}
\newcommand{\Udp}[2]{f_{p^2}^{\mathcal{U}} \big({#1},{#2}\big)}
\newcommand{\Ltp}[3]{f_{p^3}^{\mathcal{L}} \big({#1},{#2},{#3}\big)}
\newcommand{\Utp}[3]{f_{p^3}^{\mathcal{U}} \big({#1},{#2},{#3}\big)}

\newcommand{\fah}[2]{f_{\va_\mathtt{h}} \big({#1},{#2}\big)}
\newcommand{\fav}[2]{f_{\va_\mathtt{v}} \big({#1},{#2}\big)}
\begin{document}

	\title{Energy Efficiency for Massive MIMO Integrated Sensing and Communication Systems}
	\author{Huy~T.~Nguyen, \IEEEmembership{Member, IEEE}, Van-Dinh~Nguyen, \IEEEmembership{Senior Member, IEEE},
	 		Nhan~Thanh~Nguyen, \IEEEmembership{Member, IEEE},
                Nguyen Cong Luong,
                Vo-Nguyen~Quoc~Bao, \IEEEmembership{Senior Member, IEEE},
		 	Hien~Quoc~Ngo, \IEEEmembership{Fellow, IEEE},\\
                Dusit Niyato, \IEEEmembership{Fellow, IEEE}, 
                and 
                Symeon Chatzinotas, \IEEEmembership{Fellow, IEEE}
                
		\thanks{This research is supported by the project Sustainable Multifunctional Satellite Systems (SMS2) funded by Fonds National de la Recherche (FNR) under Contract C24/IS/18957132/SMS2, the Vietnam National Foundation for Science and Technology Development (NAFOSTED) under grant number 102.02-2023.43, the Research Council of Finland through 6G Flagship Program (grant no. 369116); project DIRECTION (grant no. 354901); project DYNAMICS (grant no. 367702), and project S6GRAN (grant no. 370561);  CHIST-ERA through the project PASSIONATE  (grant number 359817); the U.K. Research and Innovation Future Leaders Fellowships under Grant MR/X010635/1; a research grant from the Department for the Economy Northern Ireland under the US-Ireland R\&D Partnership Programme; Seatrium New Energy Laboratory, Singapore Ministry of Education (MOE) Tier 1 (RG87/22 and RG24/24), the NTU Centre for Computational Technologies in Finance (NTU-CCTF), and the RIE2025 Industry Alignment Fund - Industry Collaboration Projects (IAF-ICP) (Award I2301E0026), administered by A*STAR. Corresponding author: Van-Dinh Nguyen} 

\thanks{H. T. Nguyen is with Smart and Autonomous Systems Research Group, Faculty of Information Technology, School of Technology, Van Lang University, Ho Chi Minh City, 70000, Vietnam. (e-mail: huy.nt@vlu.edu.vn).} 
\thanks{V.-D. Nguyen is with the School of Computer Science and Statistics, Trinity College Dublin, Dublin 2, D02PN40, Ireland  (e-mail: dinh.nguyen@tcd.ie).}
\thanks{N. T. Nguyen is with Centre for Wireless Communications, University of Oulu, P.O.Box 4500, FI-90014, Finland (e-mail: nhan.nguyen@oulu.fi). }
\thanks{N. C. Luong is with the Phenikaa School of Computing, PHENIKAA University, Hanoi 12116, Vietnam (e-mail: luong.nguyencong@phenikaa-uni.edu.vn).}
\thanks{V.-N. Q. Bao is with Faculty of Information Technology, School of Technology, Van Lang University, Ho Chi Minh City, 70000, Vietnam. (e-mail: bao.vnq@vlu.edu.vn).}
\thanks{H. Q. Ngo is with the Centre for Wireless Innovation (CWI), Queen’s University Belfast, BT3 9DT Belfast, U.K. (e-mail: hien.ngo@qub.ac.uk).}
\thanks{Dusit Niyato is with the College of Computing and Data Science, Nanyang
Technological University, Singapore 639798 (e-mail: dniyato@ntu.edu.sg).}
\thanks{S. Chatzinotas is with the Interdisciplinary Centre for Security,
Reliability and Trust (SnT), University of Luxembourg, 1855 Luxembourg City, Luxembourg and with College of Electronics \& Information, Kyung Hee University, Yongin-si, 17104, Korea (e-mail:  symeon.chatzinotas@uni.lu).
}}
	\maketitle
	
\begin{abstract}
This paper explores the energy efficiency (EE) of integrated sensing and communication (ISAC) systems employing massive multiple-input multiple-output (mMIMO) techniques to leverage spatial beamforming gains for both communication and sensing. We focus on an mMIMO-ISAC system operating in an orthogonal frequency-division multiplexing setting with a uniform planar array, zero-forcing downlink transmission, and mono-static radar sensing to exploit multi-carrier channel diversity. By deriving closed-form expressions for the achievable communication rate and Cramér-Rao bounds (CRBs), we are able to determine the overall EE in closed-form. A power allocation problem is then formulated to maximize the system's EE by balancing communication and sensing efficiency while satisfying communication rate requirements and CRB constraints. Through a detailed analysis of CRB properties, we reformulate the problem into a more manageable form and leverage Dinkelbach's and successive convex approximation (SCA) techniques to develop an efficient iterative algorithm. A novel initialization strategy is also proposed to ensure high-quality feasible starting points for the iterative optimization process. Extensive simulations demonstrate the significant performance improvement of the proposed approach over baseline approaches. Results further reveal that as communication spectral efficiency rises, the influence of sensing EE on the overall system EE becomes more pronounced, even in sensing-dominated scenarios. Specifically, in the high $\omega$ regime of $2 \times 10^{-3}$, we observe a 16.7\% reduction in overall EE when spectral efficiency increases from $4$ to $8$ bps/Hz, despite the system being sensing-dominated.
\end{abstract}
	
\begin{IEEEkeywords}
Energy efficiency, integrated sensing and communication, massive multiple-input multiple-output, zero-forcing, orthogonal frequency-division multiplexing.
\end{IEEEkeywords}
\IEEEpeerreviewmaketitle

	\section{Introduction}
Future wireless systems, particularly in the sixth-generation mobile system (6G), demand unprecedented performance characterized by ultra-high data rates, extremely low latency and exceptional energy efficiency (EE)\cite{chowdhury20206g}. Beyond these requirements, 6G aims to integrate sensing and cognitive capabilities, leveraging massive multiple-input multiple-output (mMIMO) technology to enable simultaneous communication and radar functionalities \cite{fang2022joint}. This innovative approach, widely referred to as integrated sensing and communication (ISAC), is also known as dual-functional radar-communication or joint communication and sensing in the literature \cite{liu2022integrated, liu2022survey}.

Existing ISAC systems can be broadly categorized into three types based on their primary design objectives: Radar-centric, communication-centric, and joint design approaches~\cite{zhang2021overview}. Radar-centric methods embed digital messages into radar waveforms but often compromise communication performance. On the other hand, communication-centric approaches, which primarily use conventional communication signals, typically offer limited sensing capabilities. This work aims to strike a balance between communication and sensing energy efficiency through optimized transceiver design for mMIMO ISAC systems.

\subsection{Related Work}
Recent research has increasingly focused on optimizing communication and sensing performance in ISAC systems~\cite{nguyen2024performance, zhang2023isac, nguyen2023multiuser, nguyen2024joint, nguyen2021hierarchical, liu2020joint}. Sensing performance is typically evaluated using metrics such as beampattern, signal-to-cluster-plus-noise ratio (SCNR), and the Cramér-Rao bound (CRB). Several studies aim to minimize beampattern mismatch while maintaining reliable communication by enforcing a signal-to-interference-plus-noise ratio (SINR) threshold~\cite{nguyen2023multiuser, nguyen2024joint, johnston2022mimo}. On the other hand, beamforming designs for ISAC systems commonly adopt SCNR and SINR as performance metrics to evaluate sensing and communication capabilities, respectively~\cite{choi2024joint, zhang2023isac, wang2023optimizing}. While approaches based on beampattern or SCNR prioritize target coverage within radar transmitter main lobes, they frequently neglect the radar echo signal processing required for detection and estimation. By contrast, the CRB, which establishes a fundamental accuracy limit for parameter estimation, offers a more comprehensive evaluation of sensing performance. As a result, it has been adopted as a key metric in numerous studies~\cite{nguyen2024performance, song2023intelligent, ren2022fundamental, liu2021cramer}.       
	
The rising demand for EE has positioned it as a critical performance metric in modern wireless systems~\cite{nguyen2020nonsmooth, nguyen2017spectral}. In ISAC systems, the integration of sensing and communication functions presents unique challenges, requiring the simultaneous optimization of data transmission and environmental perception while minimizing energy consumption. Enhancing EE in these systems is vital for reducing operational costs and environmental impact, extending the lifespan of energy-constrained devices like battery-powered Internet of Things (IoT) sensors, and promoting sustainable wireless networks by lowering their carbon footprint~\cite{leong2024green, gandotra2017green}. 

The EE maximization has become a key design objective for ISAC systems, driving the development of innovative techniques for resource allocation, waveform design, and hardware architectures~\cite{zou2024energy, allu2024robust, wu2023energy, he2022energy}. In~\cite{he2022energy}, communication EE was optimized by constraining the beampattern design for sensing while maintaining the communication SINR threshold. The authors in \cite{zou2024energy} adopted a similar problem formulation but incorporated a CRB constraint to enhance sensing accuracy. Focusing on both downlink and uplink communication, \cite{allu2024robust} sought to maximize communication EE by imposing a radar rate constraint as the sensing performance metric. Meanwhile, \cite{wu2023energy} leveraged the diversity of MIMO systems to jointly minimize power consumption and optimize the selection of active and inactive transmitters, ensuring adherence to both minimum communication rate and maximum CRB constraints.

\subsection{Motivation and Main Contributions}
Existing EE designs for ISAC systems primarily focus on maximizing the EE of the communication subsystem while satisfying sensing performance constraints~\cite{allu2024robust, he2022energy, zou2024energy, hatami2025energy}. However, energy-efficient operations are essential for both subsystems, especially when they are integrated into a single platform like ISAC, where sensing functions often require significant power consumption~\cite{liu2018mu}. To the best of our knowledge, no existing work has simultaneously optimized EE for both subsystems. Moreover, despite the widespread adoption of orthogonal frequency-division multiplexing (OFDM) in practical systems~\cite{cheng2021hybrid, nguyen2023jointssp, nguyen2023multiuser}, existing EE research has largely overlooked its incorporation. Additionally, the EE design for mMIMO systems employing simple linear precoding techniques, such as zero-forcing (ZF), remains unexplored in the current literature.

In this work, we consider a monostatic ISAC system designed for point target sensing. The sensing module aims to estimate the target's parameters, i.e., azimuth and elevation angles, relative to the BS. CRB is leveraged as a fundamental lower bound on the achievable estimation accuracy for these parameters, acknowledging its direct relevance to the precision of target detection and tracking capabilities.
This work addresses these gaps by maximizing overall EE for joint communication and radar sensing in an mMIMO system. Specifically, we explore a downlink multiuser mMIMO-OFDM ISAC system featuring a large uniform planar array (UPA) transmitter with ZF precoding. ZF is selected for its ability to manage inter-user interference and its closed-form solution, which simplifies design and optimization. Meanwhile, OFDM offers flexible subcarrier allocation, facilitating efficient resource sharing between sensing and communication.
For the sensing functionality, we focus on a tracking scenario where the target is assumed to have been detected during an earlier search phase~\cite{nguyen2024performance, song2023intelligent}. Tracking demands sustained sensing accuracy and efficient resource allocation to maintain reliable performance without excessive energy consumption.
The key contributions of this work are summarized as follows:
\begin{itemize}
    \item We derive new closed-form expressions for the achievable communication rate and the CRB for the target's azimuth and elevation angles in mMIMO-OFDM systems with ZF beamforming. Existing CRB expressions are not directly applicable due to the unique considerations of the UPA antenna model, ZF beamforming, and the integrated ISAC architecture\footnote{Unlike the comprehensive model in \cite{mylonopoulos2024extended}, the CRB functions considered in this work do not account range and Doppler estimation. Incorporating these capabilities would introduce significant complexity and impose a substantial computational burden on the current framework.}.
    
    \item We formulate a novel EE maximization problem that jointly optimizes communication and sensing EE, subject to constraints on transmission power, communication rates, and CRBs. To this end, we first derive a closed-form expression for transmit power as a function of power allocation factors. The problem's objective function incorporates multiple performance metrics, including communication spectral efficiency (SE), sensing CRBs for the target's azimuth and elevation angles, and total power consumption. Each of these metrics is a complex, nonconvex function of the design variables, presenting significant challenges for power allocation.
    
    \item By leveraging Dinkelbach's and successive convex approximation (SCA) techniques, we develop a simple yet efficient iterative algorithm that is guaranteed to converge to at least a locally optimal solution satisfying the Karush-Kuhn-Tucker (KKT) conditions. This iterative algorithm also incorporates an efficient initialization strategy to ensure feasible starting points from the very first iterations.
    
    \item Numerical simulations demonstrate significant performance gains over baseline methods, highlighting that increasing the communication rate threshold reduces sensing energy efficiency, even when sensing remains the dominant factor in overall system performance.
\end{itemize}

\subsection{Paper Organization and Notation}
The rest of the paper is organized as follows. In Section \ref{sec_system_model}, we present the signaling models of ISAC systems and ZF ISAC beamforming. In Section \ref{sec_perf_analysis}, the closed-form expressions for the system communication rate and sensing CRB are derived, followed by the formulation of the overall EE. In Section \ref{sec_opt}, the EE maximization problem is formulated, followed by a detailed description of the proposed solution. Numerical results are subsequently presented and discussed in Section \ref{sec_sim}. Section \ref{sec_conclusion} draws the conclusion of the paper. 
	
\textit{Notations}: Scalars, vectors, and matrices are represented throughout the paper by lower-case, boldface lower-case, and boldface upper-case letters, respectively; $(\cdot)^*$, $(\cdot)^\T$, $(\cdot)^\H$, and $\tr{\cdot}$ denote the conjugate, transpose, the Hermitian transpose, and the trace operators, respectively; $\abs{\cdot}$ and $\norm{\cdot}$  respectively denote the modulus of a complex number and the Euclidean norm of a vector. The notation $\dot{\va}_o$ represents the derivative of $\va$ with respect to $o$, \textit{i.e.}, $\dot{\va}_o = \frac{\partial \va}{\partial o}$, $\mathcal{CN}(\mU[q], \sigma^2)$ denotes a complex normal distribution with mean $\mathbf{\mU[q]}$ and variance $\sigma^2$, and $\mean{\cdot}$ denotes the expected value of a random variable. 

\section{Signal Model}
\label{sec_system_model}
We consider a mono-static mMIMO-OFDM ISAC system, wherein a base station (BS) simultaneously transmits probing signals toward a target at a specified angle and data signals to $K$ single-antenna user equipments (UEs), as illustrated in Fig.~\ref{fig_system}. The BS employs $\Nt$ transmit and $\Nr$ receive antennas with $\Nt, \Nr \gg 1$, and $Q$ subcarriers. The $q$-th subcarrier frequency is given as $f_q = f_{\text{c}} + \frac{\text{BW}(q-Q)}{2Q}$, where $\text{BW}$ and $f_{\text{c}}$ represent the system bandwidth and center frequency, respectively. 

\begin{figure}[t]%
	\vspace{-0.5cm}\centering
	\includegraphics[scale=0.6]{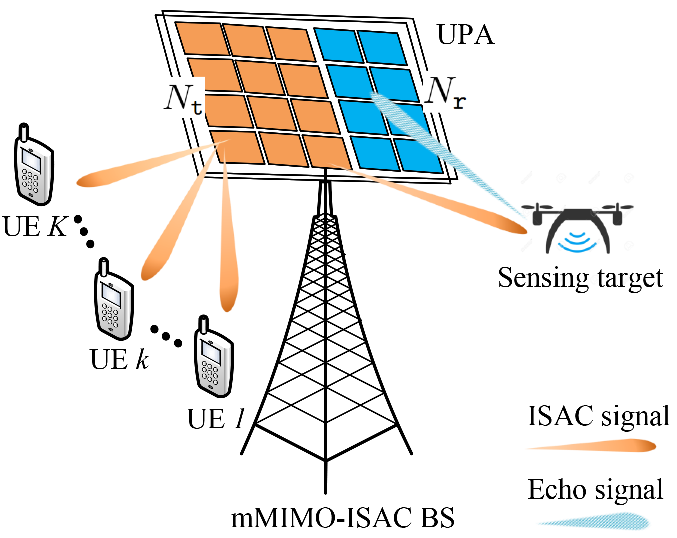}
	\caption[]{An mMIMO-ISAC system with a BS-equipped UPA transmitter ($\Nt$ transmit and $\Nr$ receive antennas).}%
	\label{fig_system}
\end{figure}

\subsection{Communication Model}
\subsubsection{Communication Subsystem}
Let $\vs_{\ell}[q] = [s_{1\ell}[q],  \ldots, s_{K\ell}[q]]^\T \in {\mathbb{C}}^{K \times 1}$ be the transmit symbol vector from the BS at the $\ell$-th time slot and $q$-th subcarrier, where $\mean{\vs_{\ell}[q] \vs_{\ell}^{\H}[q]}=\mI_K$. Then, the transmit signal matrix at the $q$-th subcarrier is given as $\mS[q] = [\vs_{1}[q], \ldots, \vs_{L}[q]] \in \setC^{K \times L}$, where $L$ is the length of the radar/communication frame. Assuming that $\{\vs_{1}[q], \ldots, \vs_{L}[q]\}$ are independent, with sufficiently large $L$, leading to
$\mS[q] \mS[q]^\H \approx L \mI_K$~\cite{liu2021cramer}. Building on a different paradigm than dedicated spatial stream designs (e.g., \cite{mylonopoulos2024extended}), our work proposes a shared stream design for both functionalities. This approach not only reduces the number of beamforming coefficients but also effectively avoids interference from the sensing stream to communication users.

We denote by $\vw_k[q]\in \mathbb{C}^{\Nt \times 1}$ and $\gamma_k[q]$  the precoding vector and power fraction allocated for the $k$-th UE at subcarrier $q$, respectively. Furthermore, we define $\mW[q] \triangleq [\vw_1[q],\ldots,\vw_K[q]] \in \mathbb{C}^{\Nt \times K}$ and {$\mD_{\gamma}[q] \triangleq \diagshort{\sqrt{\gamma_1[q]}, \ldots, \sqrt{\gamma_K[q]}} \in \setR^{K \times K}$}. On the other hand, let $\vv[q] \in \mathbb{C}^{\Nt \times 1}$ and {$\bar{\bm{\eta}}[q] = [\sqrt{ \eta_1[q]},\ldots,\sqrt{ \eta_K[q]}]^\T \in \setR^{K \times 1}$} be the precoding vector and power fraction allocated for sensing in each data stream, respectively. The downlink precoder for the transmission over the $q$-th subcarrier can be expressed as
\begin{align*}
	\mF[q] = \mW[q] \mD_{\gamma}[q] + \vv[q] \bar{\bm{\eta}}[q]^\T \in \mathbb{C}^{\Nt \times K}. \nbthis \label{eq_F}
\end{align*}
The dual-functional waveform design of the downlink precoder, as defined in \eqref{eq_F}, enhances the synergy between communication and sensing EE by jointly leveraging their shared spectral and hardware resources.
The precoders $\mW[q]$ and $\vv[q]$ are normalized such that $\mean{\tr{\mW[q] \mW^\H[q]}} = K$ and $\norm{\vv[q]}^2 = 1, \forall q$, and $\gamma_k[q] + \eta_k[q] = 1, \forall k, q$.
By \eqref{eq_F}, the dual-functional precoding vector for UE $k$, \textit{i.e.}, the $k$-th column of $\mF[q]$, is given as $\vf_k[q] = \sqrt{\gamma_k[q]} \vw_k[q] + \sqrt{\eta_k[q]} \vv[q]$.
Let $\bm{\mD}_{\xi}[q] \triangleq \diagshort{\sqrt{\xi_1[q]}, \ldots, \sqrt{\xi_K[q]}} \in \setR^{K \times K}$ be the power allocation matrix, where $\xi_k[q]$ is the power allocated to the $k$-th UE at the $q$-th subcarrier. Then, the transmitted signal vector at the $q$-th subcarrier is given as
\begin{align*}
	\vx_\ell[q] &= \mF[q] {\bm{\mD}_{\xi}[q]} \vs_\ell[q] = \sum\nolimits_{k=1}^{K} {\sqrt{\xi_k[q]}} \vf_k[q] s_{k \ell}[q]. \nbthis \label{eq_x_lq}
\end{align*}
Note that unlike the systems considered in \cite{huang2020majorcom, ma2021spatial, ma2021frac, liu2017robust, liu2022transmit, pritzker2022transmit, liu2020joint, pritzker2021transmit}, the transmit waveform $\vx_\ell[q]$ in \eqref{eq_x_lq} contains only communication symbols, and the radar waveform is ensured via beamforming design detailed in Section \ref{sec_opt}. This setting allows the communication to be free of interference from the radar subsystem \cite{cheng2021hybrid, elbir2021terahertz}.
The corresponding signal received at UE $k$ is
\begin{align*}
	y^{\mathtt{c}}_{k\ell}[q] &= {\sqrt{\xi_k[q]}} \vh_k^\H[q] \vf_k[q] s_{k\ell}[q] \\
	&\phantom{=}  + \vh_k^\H[q] \sum\nolimits_{j \neq k} {\sqrt{\xi_j[q]}} \vf_j[q] s_{j \ell}[q] + {n_{\mathtt{c}}}_{k\ell}[q], \nbthis \label{eq_y_kl}
\end{align*}
where ${n_{\mathtt{c}}}_{k\ell}[q]$ is the additive white Gaussian noise (AWGN) following the distribution $\mathcal{CN}(0, \sigmac)$; $\vh_k[q]\in \mathbb{C}^{\Nt \times 1}$ is the channel vector between the BS and UE $k$ and is modeled as $\vh_k[q] = \beta_k^{1/2} \bar{\vh}_k[q]$, where,  $\beta_k$ and $\bar{\vh}_k[q] \sim \mathcal{CN}(0, \mI_{\Nt})$ represent the large-scale and the small-scale Rayleigh fading channel coefficients, respectively \cite{mollen2016uplink}. The small-scale fading experienced by different communication links is statistically independent. In this paper, we assume that the BS has the perfect channel state information (CSI) of the communication channels,\footnote{Practically, CSI can be accurately and efficiently acquired through channel estimation at the receiver, then effectively shared with the transmitter via feedback techniques \cite{yu2016alternating}.} while the BS needs to estimate the target parameters based on the echo signals. Furthermore, this paper does not consider hardware impairments (e.g., power amplifier non-linearities, phase noise, I/Q imbalance) or estimation errors.

\subsection{Radar Subsystem}
Our radar model assumes a single point target in a free-space environment, neglecting the effects of clutter and multipath propagation. The sensing channel is considered to be a direct Line-of-Sight (LoS) path to the target, and the target's Radar Cross Section (RCS) is assumed constant and known.
Let $\mX[q] \triangleq [\vx_1[q], \ldots, \vx_L[q]] \in \setC^{\Nt \times L}$ be the overall dual-functional transmit waveform at the BS. From \eqref{eq_x_lq}, we have
\begin{align}
	\mX[q] = \mF[q]  {\bm{\mD}_{\xi}[q]} \mS[q].
	\label{eq_mX}
\end{align}
The received sensing echo signal at the BS via the $q$-th subcarrier is given as \cite{liu2021cramer, song2023intelligent}:
\begin{align*}
	\mY_{\mathtt{s}}[q] &= \alpha \mG(\theta, \phi, f_q) \mX[q] + \mN_{\mathtt{s}}[q] \in \setC^{\Nr \times L}, \nbthis \label{eq_radar_model}
\end{align*}
where $\alpha$ is the reflection coefficient, and $\mN_{\mathtt{s}}[q] \in \setC^{\Nr \times L}$ is an AWGN matrix with elements following distribution $\mathcal{CN}(0, \sigmas)$. Matrix $\mG(\theta, \phi, f_q) \in \setC^{\Nr \times \Nt}$ represents the two-way channel between the BS and the sensing target~\cite{johnston2022mimo, liu2020joint}:
\begin{align}
	\mG(\theta, \phi, f_q) &= \vb(\theta, \phi, f_q) \va^\H(\theta, \phi, f_q),\label{eq_A}
\end{align} 
where $\va(\theta, \phi, f_q) \in \setC^{\Nt \times 1}$ and $\vb(\theta, \phi, f_q) \in \setC^{\Nr \times 1}$ are the transmit and receive steering vectors, respectively. Here, we adopt the practical and general assumption that the BS is equipped with a UPA containing a large number of antennas.
Unlike a Uniform Linear Array (ULA), a UPA inherently enables two-dimensional (2D) sensing capabilities, allowing for simultaneous estimation of both azimuth and elevation angles. This 2D sensing is critical for accurately localizing targets in complex three-dimensional environments and for precise beamforming to both communication users and sensing targets. The deployment of UPA is practical to support advanced applications requiring comprehensive spatial awareness, such as autonomous driving, drone navigation, and smart environment monitoring.
Our analysis and designs are also applicable to ULA, which is a specific configuration of the UPA. Accordingly, $\theta$ and $\phi$ are the azimuth and elevation angles of the target relative to the BS with $\theta \in \left[-\pi, \pi\right]$, and $\phi \in \left[-\frac{\pi}{2}, \frac{\pi}{2}\right]$, and $f_q$ is the frequency for the $q$-th subcarrier. In the subsequent analysis, we omit the parameters $(\theta, \phi, f_q)$ in $\va$ and $\vb$ for ease of exposition. Assuming half-wavelength antenna spacing, $\va$ can be expressed as
\begin{align*}
	\va[q] = \ah[q] \otimes \av[q], \nbthis \label{eq_at}
\end{align*} 
where $\ah[q]\!\triangleq\!\Big[\fah{q}{1},\fah{q}{2}, \ldots, \fah{q}{\Nth}\Big]^\T$ and $\av[q]\triangleq\Big[\fav{q}{1},\fav{q}{2}, \ldots, \fav{q}{\Ntv}\Big]^\T$ are the array response vectors corresponding to the horizontal and vertical dimensions of the UPA, respectively, with $\fah{q}{n} \triangleq e^{j\pi \frac{f_q}{f_c} (n-1) \sin(\theta) \sin(\phi)}$ and $\fav{q}{n} \triangleq e^{j\pi \frac{f_q}{f_c} (n-1) \cos(\phi)}$. $\Nth$ and $\Ntv$ represent the number of antennas in the horizontal and vertical dimensions, respectively. The size of the UPA and the total number of antennas in the UPA are respectively given by $\Nth \times \Ntv$ and $\Nt = \Nth \Ntv$. The receive response vector $\vb$ is modeled similarly.
For the purpose of analysis, we assume perfect synchronization between the communication and sensing modules in terms of time and frequency. This implies that both functionalities share a common timing reference and operate on perfectly aligned frequency bands, which is essential for coherent processing and accurate parameter estimation. While practical ISAC deployments would necessitate robust synchronization protocols, this assumption allows us to focus on the fundamental trade-offs and EE optimization between communication and sensing without introducing additional synchronization-related complexities.

\subsection{ZF ISAC Beamforming}
From \eqref{eq_F}, \eqref{eq_mX}--\eqref{eq_A}, we rewrite \eqref{eq_radar_model} equivalently as

\begin{align*}
	\mY_{\mathtt{s}}[q]
	&= \alpha \vb[q] \va^\H[q] \vv[q] \bm{\eta}^\T[q] {\bm{\mD}_{\xi}[q]}  \mS[q] \\
	&\quad + \alpha \vb[q] \va^\H[q] \mW[q] \mD_{\gamma}[q] {\bm{\mD}_{\xi}[q]} \mS[q] + \vn_{\mathtt{s}}[q].
\end{align*}
It is true that for an arbitrary $\mW[q] \mD_{\gamma}[q]$, setting $\vv[q] = \va[q] {/\sqrt{\Nt}}$
maximizes the received echo signal power at the BS and satisfies $\norm{\vv[q]}^2 = 1$. To ensure the inter-user interference cancellation capability, the BS employs the ZF precoder for communication, as derived in the following lemma.
\begin{lemma} \label{lm_ZF_precoder}
	The perfect ZF cancellation for the communication precoder is expressed as:
	\begin{align*}
		\mW[q] = {\alpha_{\mathtt{ZF}}[q]} \mH[q] \left(\mH^\H[q] \mH[q]\right)^{-1} = {\alpha_{\mathtt{ZF}}[q]} \mH^{\dagger}[q],\, \forall q, \nbthis \label{eq_W_linear}
	\end{align*}
where ${\alpha_{\mathtt{ZF}}[q]} = \sqrt{\frac{K(\Nt-K)}{\tr{\mD_{\beta}^{-1}}}}$ is the normalization factor to ensure $\mean{\tr{\mW[q] \mW^\H[q]}} = K, \forall q$.
\end{lemma}
\begin{proof}
See Appendix \ref{appd_proof_ZF}. \epr
\end{proof}

Combining \eqref{eq_W_linear} and $\vv[q] = \va[q] {/\sqrt{\Nt}}$, the overall dual-function precoder now can be given by
\begin{align*}
	\mF[q] &=  {\alpha_{\mathtt{ZF}}[q]} \mH^{\dagger}[q]  \mD_{\gamma}[q] + {\frac{1}{\sqrt{\Nt}}} \va[q] \bar{\bm{\eta}}^\T[q], \nbthis \label{eq_F_MRT} 
\end{align*}
and thus,
\begin{align*}
	\vf_k[q] &= {\alpha_{\mathtt{ZF}}[q]} \sqrt{\gamma_k[q]} \check{\vh}_k[q] + {\frac{1}{\sqrt{\Nt}}} \sqrt{\eta_k[q]} \va[q], \nbthis \label{eq_fkq_bf}
\end{align*}
where $\check{\vh}_k[q]$ in \eqref{eq_fkq_bf} is the $k$-th column of $\mH^{\dagger}[q]$.

\section{Energy Efficiency}
\label{sec_perf_analysis}
We now derive closed-form expressions for transmit power, communication SE, and sensing CRBs, which are subsequently used to formulate the EE.

\subsection{Power Consumption Model}
Let us begin by formulating the BS's transmit power for both communication and sensing over all the subcarriers. 
\begin{lemma} \label{lemma_power}
	The transmit power at the BS is given by
	\begin{align*}
		{P_{\mathtt{TX}}(\bm{\Omega}) = \sum\nolimits_{q=1}^Q \bm{\xi}^\T[q] \left( \bar{\mD}_{\beta} \bm{\gamma}[q] + \bm{\eta}[q] \right),} \nbthis \label{eq_Ptx}
	\end{align*}
	where {$\bm{\Omega} \triangleq \{ \xi_k[q], \gamma_k[q], \eta_k[q] \}_{\forall k, q}$ encompasses all the decision variables}, and $\bar{\mD}_{\beta} \triangleq \frac{K}{\tr{\mD_{\beta}^{-1}}} \mD_{\beta}^{-1}$.
\end{lemma}

\begin{proof}
	See Appendix \ref{appd_proof_power}.\epr 
\end{proof}

Taking the static and dynamic power components\footnote{Given that dynamic circuit energy consumption is independent of amplifier limits, maximizing EE under a minimum overall throughput requirement is a reasonable approach to achieve high EE along with a desirable SE, and guaranteed QoS for each UE.} into consideration, the total power consumption is computed as
\begin{align*}
	P_{\mathtt{tot}}(\bm{\Omega}) = \varrho^{-1} P_{\mathtt{TX}}(\bm{\Omega}) + P_0 + \epsilon R(\bm{\Omega}), \nbthis \label{eq_total_power}
\end{align*}
where $\varrho$ is the power amplifier efficiency, $P_0$ is the total power consumed by the circuitries in the RF chains, power supply and cooling systems, etc. $\epsilon R(\bm{\Omega})$ captures the traffic-dependent power consumption, where $\epsilon>0$ is the dynamic power consumption per unit data rate, as detailed in \cite{he2022energy}. The system sum rate, denoted as $\SEc$, and EE are derived next.

\subsection{Communication EE}
We denote by $\SEc (\bm{\Omega}) \triangleq \sum_{k=1}^K \SEck (\bm{\Omega})$ the communication SE, where $\SEck (\bm{\Omega})$ is the SE of UE $k$. In \eqref{eq_y_kl}, assume that symbol $s_{k\ell}[q]$ is intended for UE $k$ via subcarrier $q$ at time slot $\ell$. $\SEck (\bm{\Omega})$ is given by
\begin{align*}
	\!\!\SEck (\bm{\Omega}) \!\!=\!\!  \sum_{q=1}^Q\! \log_2 \!\left(\! 1\! +\! \frac{{\xi_k[q]} \abs{\vh_k^\H[q] \vf_k[q]}^2}{\sum_{j \neq k} {\xi_j[q]}  \abs{ \vh_k^\H[q] \vf_{j}[q]}^2 \!+\! \sigmac}\right). \nbthis \label{eq_rate0}
\end{align*}
The EE of the communication subsystem is thus defined as
\begin{align}\label{def_EEc}
    \mathtt{EE}_{\mathtt{c}}(\bm{\Omega}) = \frac{\SEc (\bm{\Omega})}{P_{\mathtt{tot}}(\bm{\Omega})}\ [\text{bits/J/Hz}].
\end{align}
We recall that $\vf_k[q]$ is defined in \eqref{eq_fkq_bf}. With ZF precoding, it holds that $\vh_k^\H[q] \check{\vh}_i[q] = 1$ when $i = k$, and $\vh_k^\H[q] \check{\vh}_i[q] = 0$ otherwise. Thus, we obtain
\begin{align*}
		\!\!\!\!\!\abs{\vh_k^\H {\vf}_i[q]}^2 
		\!\!=\!\! \begin{cases}
			\begin{array}{ll}
				\!\!\!\!\! {\alpha_{\mathtt{ZF}}^2[q]} \gamma_i[q] \!+\! {\frac{1}{\Nt}} \eta_i[q] \abs{\vh_k^\H [q] \va[q]}^2 \\
				\!\!\!+ 2  {\frac{{\alpha_{\mathtt{ZF}}[q]}}{\sqrt{\Nt}}} \sqrt{\gamma_i[q] \eta_i[q]} \re{\vh_k^\T [q] \va^*[q]},\!\!\! &\!\!\!\text{$k \!=\! i$}\\
				\!\!\!\!\!{\frac{1}{\Nt}} \eta_i[q] \abs{\vh_k^\H [q] \va[q]}^2,\!\!\!\!\!\! &\!\!\!\text{$k \!\neq\! i$}.
			\end{array}
		\end{cases}
		\nbthis \label{eq_hkfmrti}
\end{align*}
As a result, $\SEc (\bm{\Omega})$ is given by \eqref{eq_rate1}.
\begin{figure*}
	\begin{align*}
		\SEc (\bm{\Omega}) = \sum_{k=1}^K \sum_{q=1}^Q \log_2 \left( 1 + \frac{{\xi_k[q]} \left( {\alpha_{\mathtt{ZF}}^2[q]} \gamma_k[q] + {\frac{1}{\Nt}} \eta_k[q] \abs{\vh_k^\H [q] \va[q]}^2 + 2  {\frac{ {\alpha_{\mathtt{ZF}}[q]}}{\sqrt{\Nt}}} \sqrt{\gamma_k[q] \eta_k[q]} \re{\vh_k^\T [q] \va^*[q]}\right)}{ {\frac{1}{\Nt}}  \sum_{j \neq k} {\xi_j[q]} \eta_j[q] \abs{\vh_k^\H [q] \va[q]}^2  + \sigmac}\right). \nbthis \label{eq_rate1}
	\end{align*}
	\hrule
\end{figure*}

\subsection{Sensing EE}
The definition of sensing EE remains an open question in the literature. In \cite{zou2024energy}, it is defined as the ratio of the inverse CRB to the total power consumption, considering only the CRB for the azimuth angle. In this work, we extend this definition by incorporating the sum of the inverse CRBs for both azimuth and elevation angles in the numerator. Consequently, the sensing EE is formulated as
\begin{align*}
	{\mathtt{EE}_{\mathtt{s}}(\bm{\Omega}) = \frac{ \mathtt{CRB}^{-1}_{\theta}(\bm{\Omega}) + \mathtt{CRB}^{-1}_{\phi}(\bm{\Omega}) }{P_{\mathtt{tot}}(\bm{\Omega})}}\ \text{[1/rad$^2$/W]}, \nbthis \label{eq_EE}
\end{align*}	
where $\mathtt{CRB}_{\theta}(\bm{\Omega})$ and $ \mathtt{CRB}_{\phi}(\bm{\Omega})$ are the CRBs for $(\theta, \phi)$, respectively, obtained in Theorem \ref{theo_CRB}. When using the sensing EE defined in \eqref{eq_EE} as the design objective for the sensing subsystem, it optimizes sensing performance while minimizing the total energy consumption, thereby improving the overall system's EE. 
\begin{theorem}
	\label{theo_CRB}
    The closed-form expressions of CRBs for $\theta$ and $\phi$ are derived as follows:
		\begin{align*}
			\mathtt{CRB}_{\theta}(\bm{\Omega}) \!&=\! \Big(\Ttt(\bm{\Omega}) - \frac{\Ttp^2(\bm{\Omega})}{\Tpp(\bm{\Omega}) - \Tpap(\bm{\Omega})}\!\Big)^{-1}\!\!\!\!\!,\ \nbthis \label{eq_CRB_theta_result}\\
			\mathtt{CRB}_{\phi}(\bm{\Omega}) \!&=\! \Big(\Tpp(\bm{\Omega}) \!-\! \Tpap(\bm{\Omega}) \!-\! \frac{\Ttp^2(\bm{\Omega})}{\Ttt(\bm{\Omega})}\Big)^{-1}\!\!\!\!\!\!, \nbthis \label{eq_CRB_phi_result}
		\end{align*}	
where $\Tpap(\bm{\Omega}) \triangleq \Tpa(\bm{\Omega}) \Taa^{-1}(\bm{\Omega}) \Tpa^\T(\bm{\Omega})$, and $\Ttt, \Ttp, \Tta, \Tpp, \Tpa$, and $\Taa$ are given as 
	\begin{align*}
			&\Ttt(\bm{\Omega}) = \bar{\kappa} \abs{\alpha}^2 \sum_{q=1}^Q \bm{\xi}^\T[q] \left( \hat{\mD}_{\beta}^{\theta \theta}[q] \bm{\gamma}[q]  + \bm{\eta}[q] c_{\theta \theta}[q] \right) 
			, \nbthis \label{eq_result_Ttt} \\
			&\Ttp(\bm{\Omega}) = \bar{\kappa} \abs{\alpha}^2 \sum_{q=1}^Q \bm{\xi}^\T[q] \left( \hat{\mD}_{\beta}^{\theta \phi}[q] \bm{\gamma}[q] +  \bm{\eta}[q] c_{\theta \phi}[q] \right), \nbthis \label{eq_result_Ttp}\\
			&\Tta(\bm{\Omega}) \!=\! \bar{\kappa} \re{\!\!\alpha^*\!\! \sum_{q=1}^Q\! \bm{\xi}^\T[q] \!\!  \left( \hat{\mD}_{\beta}^{\theta}[q] \bm{\gamma}[q]\!\! +\!\! \bm{\eta}[q] c_{\theta}[q] \right) \!\! [1,j]\!\!} , \nbthis \label{eq_result_Tta}\\
			&\Tpp(\bm{\Omega}) = \bar{\kappa} \abs{\alpha}^2 \sum_{q=1}^Q  \bm{\xi}^\T[q] \left(  \hat{\mD}_{\beta}^{\phi \phi}[q] \bm{\gamma}[q] + \bm{\eta}[q] c_{\phi \phi}[q] \right), \nbthis \label{eq_result_Tpp} \\
			&\Tpa(\bm{\Omega}) \!=\! \bar{\kappa} \re{\!\!\alpha^*\!\! \sum_{q=1}^Q\! \bm{\xi}^\T[q] \!\left(\! \hat{\mD}_{\beta}^{\phi}[q] \bm{\gamma}[q] \!\!+\!\! \bm{\eta}[q] c_{\phi}[q]\! \right)\! [1,j]\!\!}, \nbthis \label{eq_result_Tpa}\\
			&\Taa(\bm{\Omega}) = \bar{\kappa} \sum_{q=1}^Q \bm{\xi}^\T[q] \left(  \hat{\mD}_{\beta}[q] \bm{\gamma}[q] + \bm{\eta}[q] c_{2}[q] \right) \mI_2, \nbthis \label{eq_result_Taa}
		\end{align*}
		respectively. Herein, $x, y \in \{\theta, \phi\}$, $\bar{\kappa} \triangleq \frac{\kappa}{ \Nt}$ and $\hat{\mD}_{\beta}^{xy}[q] \triangleq \tr{\dot{\mG}^\H_{x}[q] \dot{\mG}_{y}[q]} \bar{\mD}_{\beta}$, $\hat{\mD}_{\beta}^{x}[q] \triangleq \tr{\dot{\mG}^\H_{x}[q] \mG[q]} \bar{\mD}_{\beta}$, $\hat{\mD}_{\beta}[q] \triangleq \tr{\mG^\H[q] \mG[q]} \bar{\mD}_{\beta}$, $c_{xy}[q] \triangleq \va^\H[q] \dot{\mG}^\H_{x}[q] \dot{\mG}_{y}[q] \va[q]$, $c_{x}[q] \triangleq \va^\H[q] \dot{\mG}^\H_{x}[q] \mG[q] \va[q]$, $c[q] \triangleq \va^\H[q] \mG^\H[q] \mG[q] \va[q]$ are constants with respect to the variable set $\bm{\Omega}$. 
\end{theorem}
\begin{proof}
	See Appendix \ref{appd_proof_CRB}.\epr
\end{proof}

\section{Communication and Sensing Power Allocation for EE Maximization} 
\label{sec_opt}
\subsection{Problem Formulation}
Our objective is to maximize the EE of the mMIMO-OFDM ISAC system by optimizing power allocation between communication and sensing, as well as among communication UEs. The EE maximization problem is formulated as follows:
\begin{subequations}\label{ori_probOSA}
	\begin{IEEEeqnarray} {rcl}
		& \underset{\substack{\bm{\Omega}}}{\textrm{maximize}}  & \quad  \mathtt{EE}_{\mathtt{c}}(\bm{\Omega}) + \omega \mathtt{EE}_{\mathtt{s}}(\bm{\Omega})
		\label{eq::probOSA::obj_func_OSA} \\
		&\mathtt{s.t.} 
		&\quad \SEck(\bm{\Omega}) \geq \SEth,\ \forall k,  \label{cons_comm_OSA} \\
		&& \quad \mathtt{CRB}_{\theta}(\bm{\Omega}) \leq \mathtt{CRB}_{\theta}^0,  \label{eq::probOSA::cons_CRB_theta} \\
		& &\quad\mathtt{CRB}_{\phi}(\bm{\Omega}) \leq \mathtt{CRB}_{\phi}^0,  \label{eq::probOSA::cons_CRB_phi} \\
		& &\quad \sum_{q=1}^Q \bm{\xi}^\T[q] \left( \bar{\mD}_{\beta} \bm{\gamma}[q] + \bm{\eta}[q] \right) \leq P_{\mathtt{max}}, \label{eq::probOSA::cons_power_MRT}\\
		& &\quad\gamma_k[q] + \eta_k[q] = 1, \forall k, q,  \label{eq_total_power_0} \\
		& &\quad \xi_k[q]\geq 0,\gamma_k[q] \geq 0,\eta_k[q] \geq 0, \forall k,q, \label{eq_xi_power}\quad%
	\end{IEEEeqnarray}
\end{subequations}
where $\omega$ represents a weighting coefficient, i.e., trade-off parameter, used to balance the EE between the communication and sensing subsystems or to prioritize one over the other in the weighted sum objective function \eqref{eq::probOSA::obj_func_OSA}. Furthermore, it is important to note that $\mathtt{EE}_{\mathtt{c}}$ and $\mathtt{EE}_{\mathtt{s}}$ have different units, as defined in \eqref{def_EEc} and \eqref{eq_EE}, respectively. Consequently, their weighted sum in \eqref{eq::probOSA::obj_func_OSA} serves as a dimensionless design objective metric, representing a normalized trade-off between the two EE components. 
Constraint \eqref{cons_comm_OSA} guarantees that the SE achieved by all the UEs exceeds a threshold $\SEth$. Constraints  \eqref{eq::probOSA::cons_CRB_theta} and \eqref{eq::probOSA::cons_CRB_phi} limit CRBs for $(\theta, \phi)$ to ceilings $\mathtt{CRB}_{\theta}^0$ and $\mathtt{CRB}_{\phi}^0$, respectively. Constraint \eqref{eq::probOSA::cons_power_MRT} caps the transmit power at a predefined value $P_{\mathtt{max}}$. Constraint \eqref{eq_total_power_0} guarantees that the power fractions are allocated to either communication or sensing functionalities. The problem has a nonconcave and fractional objective function \eqref{eq::probOSA::obj_func_OSA} and nonconvex constraints \eqref{cons_comm_OSA}--\eqref{eq::probOSA::cons_power_MRT}, making it inherently non-convex and challenging. To address this, we employ Dinkelbach's and SCA techniques as outlined next.

\subsection{Convex Approximation for \eqref{ori_probOSA}}
In the following step, we will convexify the nonconcave objective function and nonconvex constraints of \eqref{ori_probOSA}. In particular, we will leverage the preliminary results established in \eqref{inequality_xy_geq}-\eqref{inequality_txy_leq}, where  $\Ldp{\xi_k[q]}{\gamma_k[q]}$, $\Ldp{\xi_k[q]}{\eta_k[q]}$, $\Ltp{\xi_k[q]}{\sqrt{\gamma_k[q]}}{\sqrt{\eta_k[q]}}$ denote concave lower bounds of $\xi_k[q] \gamma_k[q]$, $\xi_k[q] \eta_k[q]$ and $\xi_k[q] \sqrt{\gamma_k[q] \eta_k[q]}$, respectively, while $\Udp{\xi_k[q]}{\gamma_k[q]}$, $\Udp{\xi_k[q]}{\eta_k[q]}$, $\Utp{\xi_k[q]}{\sqrt{\gamma_k[q]}}{\sqrt{\eta_k[q]}}$ represent their corresponding convex upper bounds. 

\subsubsection{Convexifying Objective Function \eqref{eq::probOSA::obj_func_OSA}} 
To tackle the non-convexity of \eqref{eq::probOSA::obj_func_OSA}, which arises from the ratio of concave and convex functions, we adopt the inner approximation framework~\cite{Beck:JGO:10}. Each iteration of Dinkelbach's method then requires solving:
\begin{align*}
	\!\!\! \underset{\substack{\bm{\Omega}}}{\textrm{max}} \{\SEc(\bm{\Omega})\!+\! \omega( \mathtt{CRB}^{-1}_{\theta}(\bm{\Omega}) \!+\! \mathtt{CRB}^{-1}_{\phi}(\bm{\Omega}) )  \!-\! \tau P_{\mathtt{tot}}(\bm{\Omega}) \} \nbthis \label{eq_obj_dkb}
\end{align*}
where $\tau$ is updated iteratively to drive the optimal value of \eqref{eq_obj_dkb} to zero. In each iteration, the objective function becomes more tractable compared to the original fractional form in \eqref{ori_probOSA}. However, it remains challenging due to the involvement of various communication, sensing, and power metrics, specifically $\SEc(\bm{\Omega})$, $\mathtt{CRB}^{-1}_{\theta}(\bm{\Omega})$, $\mathtt{CRB}^{-1}_{\phi}(\bm{\Omega})$, and $P_{\mathtt{tot}}(\bm{\Omega})$. These metrics are complex non-convex functions with tightly coupled power variables in $\bm{\Omega}$. In the following, we address each of these terms in detail.

\paragraph{Convex Approximation of $\mathtt{EE}_{\mathtt{c}}(\bm{\Omega})$}
Recall that $\SEc \triangleq
\sum_{k=1}^K \SEck(\bm{\Omega}) =  \sum_{k=1}^K \sum_{q=1}^Q\! \log_2 \!\left(\! 1\! +\! \frac{\mathcal{R}_{\mathtt{Nu}, k} (\bm{\Omega})[q]}{\mathcal{R}_{\mathtt{De}, k} (\bm{\Omega})[q]}\right)$ is given in \eqref{eq_rate1}, where
\begin{align*}
	\mathcal{R}_{\mathtt{Nu}, k} (\bm{\Omega})[q] &\triangleq {\xi_k[q]} \left( {\alpha_{\mathtt{ZF}}^2[q]} \gamma_k[q] + {\frac{1}{\Nt}} \eta_k[q] \abs{\vh_k^\H [q] \va[q]}^2 \right. \\
	&\hspace{0.2cm} \left. + 2  {\frac{ {\alpha_{\mathtt{ZF}}[q]}}{\sqrt{\Nt}}} \sqrt{\gamma_k[q] \eta_k[q]} \re{\vh_k^\T [q] \va^*[q]}\right),  \nbthis \label{eq_ori_RNu} \\
	\mathcal{R}_{\mathtt{De}, k} (\bm{\Omega})[q] &\triangleq {\frac{1}{\Nt}}  \sum_{j \neq k} {\xi_j[q]} \eta_j[q] \abs{\vh_k^\H [q] \va[q]}^2  + \sigmac. \nbthis \label{eq_ori_RDe}
\end{align*}
 Let {$\bm{\Omega}^{(i)} = \{ \xi_k^{(i)}[q], \gamma_k^{(i)}[q], \eta_k^{(i)}[q] \}_{\forall k, q}$} denote the feasible point set found at iteration $i$. By \cite[Eq. (74)]{NasirTWC21}, we first find the global concave lower bound of $\ln\big(1+x/y\big)$ as
\begin{align*}
	\SEck(\bm{\Omega}) &\geq \frac{1}{\ln2} \sum_{q=1}^Q \Big[ {A_{k}^{(i)}}[q] - \frac{{B_{k}^{(i)}}[q]}{\mathcal{R}_{\mathtt{Nu}, k} (\bm{\Omega})[q]} \nonumber \\
    &\qquad\quad - {C_{k}^{(i)}}[q] \mathcal{R}_{\mathtt{De}, k} (\bm{\Omega})[q] \Big] 
	\triangleq	\SEck^{(i)}(\bm{\Omega}), \nbthis \label{eq_approx_Rk}
\end{align*}
where 
\begin{align*}
	{A_{k}^{(i)}}[q] &\triangleq \ln\Big(1+\frac{\mathcal{R}_{\mathtt{Nu}, k} (\bm{\Omega}^{(i)})[q]}{\mathcal{R}_{\mathtt{De}, k} (\bm{\Omega}^{(i)})[q]}\Big) \\
	&\hspace{2cm} + 2 \frac{\mathcal{R}_{\mathtt{Nu}, k} (\bm{\Omega}^{(i)})[q]}{\mathcal{R}_{\mathtt{Nu}, k} (\bm{\Omega}^{(i)})[q]+\mathcal{R}_{\mathtt{De}, k} (\bm{\Omega}^{(i)})[q]},\nonumber\\
	{B_{k}^{(i)}}[q] &\triangleq \frac{(\mathcal{R}_{\mathtt{Nu}, k} (\bm{\Omega}^{(i)})[q])^2}{\mathcal{R}_{\mathtt{Nu}, k} (\bm{\Omega}^{(i)})[q]+\mathcal{R}_{\mathtt{De}, k} (\bm{\Omega}^{(i)})[q]},\nonumber\\
	{C_{k}^{(i)}}[q] &\triangleq  \frac{\mathcal{R}_{\mathtt{Nu}, k} (\bm{\Omega}^{(i)})[q]}{(\mathcal{R}_{\mathtt{Nu}, k} (\bm{\Omega}^{(i)})[q]+\mathcal{R}_{\mathtt{De}, k} (\bm{\Omega}^{(i)})[q])\mathcal{R}_{\mathtt{De}, k} (\bm{\Omega}^{(i)})[q]}.
\end{align*}
Next,  substituting $\Ldp{\xi_k[q]}{\gamma_k[q]}$, $\Ldp{\xi_k[q]}{\eta_k[q]}$ and $\Ltp{\xi_k[q]}{\sqrt{\gamma_k[q]}}{\sqrt{\eta_k[q]}}$ into \eqref{eq_ori_RNu}, we obtain the following approximation of $\mathcal{R}_{\mathtt{Nu}, k} (\bm{\Omega})[q]$
{\small\begin{align*}
	&\mathcal{R}_{\mathtt{Nu}, k} (\bm{\Omega})[q] \\
    &\geq {\alpha_{\mathtt{ZF}}^2[q]} \Ldp{\xi_k[q]}{\gamma_k[q]} + {\frac{1}{\Nt}} \abs{\vh_k^\H [q] \va[q]}^2 \Ldp{\xi_k[q]}{\eta_k[q]}  \\
	&\hspace{0.5cm} + 2  {\frac{ {\alpha_{\mathtt{ZF}}[q]}}{\sqrt{\Nt}}} \Ltp{\xi_k[q]}{\sqrt{\gamma_k[q]}}{\sqrt{\eta_k[q]}} \re{\vh_k^\T [q] \va^*[q]} \\ 
    &\triangleq \tilde{\mathcal{R}}_{\mathtt{Nu}, k} (\bm{\Omega})[q]. \nbthis \label{RNu_lwb}
\end{align*}}Furthermore, the nonconvex term $\mathcal{R}_{\mathtt{De}, k} (\bm{\Omega})[q]$ can be converted to a convex form by deriving a convex upper bound $\Udp{\xi_j[q]}{\eta_j[q]}$,  such as
\begin{align*}
	\mathcal{R}_{\mathtt{De}, k} (\bm{\Omega})[q] &\leq {\frac{1}{\Nt}}  \sum_{j \neq k} \Udp{\xi_j[q]}{\eta_j[q]} \abs{\vh_k^\H [q] \va[q]}^2  + \sigmac \\
	&\triangleq \tilde{\mathcal{R}}_{\mathtt{De}, k} (\bm{\Omega})[q]. \nbthis \label{RDe_upb}
\end{align*}
Substituting the convex bounds \eqref{RNu_lwb} and \eqref{RDe_upb} into \eqref{eq_approx_Rk} yields
\begin{align*}
	&\SEck(\bm{\Omega}) \geq \SEck^{(i)}(\bm{\Omega}) \nonumber \\
	&\geq \frac{1}{\ln2} \sum_{q=1}^Q \Big[ {A_{k}^{(i)}}[q] - \frac{{B_{k}^{(i)}}[q]}{\tilde{\mathcal{R}}_{\mathtt{Nu}, k} (\bm{\Omega})[q]} - {C_{k}^{(i)}}[q] \tilde{\mathcal{R}}_{\mathtt{De}, k} (\bm{\Omega})[q] \Big] \\
	&\triangleq	\wtSEck^{(i)}(\bm{\Omega}),\nbthis \label{eq_approx_Rk_new}
\end{align*}
satisfying $\SEck^{(i)}(\bm{\Omega}) = \SEck^{(i)}(\bm{\Omega}^{(i)})$. 

\paragraph{Convex Approximation of $\mathtt{CRB}_{\theta}^{-1} (\bm{\Omega})$}  Recalling $\mathtt{CRB}_{\theta} (\bm{\Omega})$ in \eqref{eq_CRB_theta_result}, the first term $\Ttt(\bm{\Omega})$ of its inverse can be rewritten by
\begin{IEEEeqnarray}{cl}
		&\Ttt(\bm{\Omega})= \bar{\kappa} \abs{\alpha}^2 \!\sum_{q=1}^Q \!\left( \!\bm{\xi}^\T[q] \hat{\mD}_{\beta}^{\theta \theta}[q] \bm{\gamma}[q]  \!+\! \bm{\xi}^\T[q] \bm{\eta}[q] c_{\theta \theta}[q] \right) \nonumber\\
		&=  \bar{\kappa} \abs{\alpha}^2\! \sum_{q=1}^Q \!\sum_{k=1}^{K} \!\left(\! \xi_k[q] \gamma_k[q]\! \left[ \hat{\mD}_{\beta}^{\theta \theta}[q]  \right]_{kk} \!\!\!\!\!+\! \xi_k[q] \eta_k[q] c_{\theta \theta}[q] \right).\qquad  \label{tau_theta_theta}
\end{IEEEeqnarray}
From \eqref{tau_theta_theta}, a concave lower bound for $\Ttt(\bm{\Omega})$ is achieved by repurposing $\Ldp{\xi_k[q]}{\gamma_k[q]}$ and $\Ldp{\xi_k[q]}{\eta_k[q]}$ that yields
\begin{align*}
    \!\!\!\Ttt(\bm{\Omega}) 
    &\geq  \bar{\kappa} \abs{\alpha}^2 \sum_{q=1}^Q \sum_{k=1}^{K} \left( \Ldp{\xi_k[q]}{\gamma_k[q]} \left[ \hat{\mD}_{\beta}^{\theta \theta}[q]  \right]_{kk} \right. \\
    &\hspace{1.0cm} \left. + \Ldp{\xi_k[q]}{\eta_k[q]} c_{\theta \theta}[q] \right) \triangleq \Ttt^{(i)}(\bm{\Omega}). \nbthis \label{tau_theta_theta_lwb}
\end{align*}		
Next, we introduce the auxiliary variables $\{\slvar{A}{x}{}, \slvar{A}{y}{}, \slvar{A}{z}{}\}$ to replace the non-convex term $\frac{\Ttp^2(\bm{\Omega})}{\Tpp(\bm{\Omega}) - \Tpap(\bm{\Omega})}$ by
\begin{align*}
	\frac{\Ttp^2(\bm{\Omega})}{\Tpp(\bm{\Omega}) - \Tpap(\bm{\Omega})} \overset{(a)}{\leq} \frac{\slvar{A}{x}{2}}{\slvar{A}{y}{}} \overset{(b)}{\leq} \slvar{A}{z}{}, \nbthis \label{crlb_soc}
\end{align*}
which results in 
\begin{align*}
	\mathtt{CRB}^{-1}_{\theta}(\bm{\Omega}) \geq \Ttt^{(i)}(\bm{\Omega}) - \slvar{A}{z}{} \triangleq \Big[\mathtt{CRB}^{-1}_{\theta}\Big]^{(i)}(\bm{\Omega}). \nbthis \label{eq_CRB_theta_inv_lwb}
\end{align*}
By employing the second-order cone (SOC) constraint transformation, $(b)$ can be rewritten equivalently as
\begin{align*}
	\Big\|\Big[ 2 \slvar{A}{x}{};\; \Big(\slvar{A}{y}{} - \slvar{A}{z}{} \Big)\!\Big]\!\Big\|_2\! \leq \!\Big(\slvar{A}{y}{} + \slvar{A}{z}{} \Big), \nbthis \label{crlb_soc_b} 
\end{align*}
with
\begin{align*}
    \Ttp(\bm{\Omega}) \leq \slvar{A}{x}{}, \text{and }
		\Tpp(\bm{\Omega}) - \Tpap(\bm{\Omega}) \geq \slvar{A}{y}{}. \nbthis \label{crlb_soc_a}
\end{align*} 
Since the left-hand side (LHS) of \eqref{crlb_soc_a} is still nonconvex, it is necessary to convexify it.  The function $\Ttp(\bm{\Omega})$ in \eqref{eq_result_Ttp} can be rewritten as
$\Ttp(\bm{\Omega})= \bar{\kappa} \abs{\alpha}^2\sum_{q=1}^Q \sum_{k=1}^{K}\left( \xi_k[q] \gamma_k[q]\left[\hat{\mD}_{\beta}^{\theta \phi}[q]  \right]_{kk}+\xi_k[q] \eta_k[q] c_{\theta \phi}[q] \right)$, which is upper bounded as
\begin{align*}
	\!\!\!\!\!\Ttp(\bm{\Omega})
	&\leq  \bar{\kappa} \abs{\alpha}^2 \sum_{q=1}^Q \sum_{k=1}^{K} \left( \Udp{\xi_k[q]}{\gamma_k[q]} \left[ \hat{\mD}_{\beta}^{\theta \phi}[q]  \right]_{kk} \right. \\
	&\hspace{1.5cm} \left. + \Udp{\xi_k[q]}{\eta_k[q]} c_{\theta \phi}[q] \right) \triangleq \Ttp^{(i)}(\bm{\Omega}). \nbthis \label{tau_theta_phi_upb}
\end{align*}
The first constraint in \eqref{crlb_soc_a} is iteratively replaced by the following convex constraint:
\begin{align*}
	\Ttp^{(i)}(\bm{\Omega}) \leq \slvar{A}{x}{}. \nbthis \label{crlb_soc_a1_approx}
\end{align*}
Similar to \eqref{tau_theta_theta_lwb} in approximating $\Ldp{\xi_k[q]}{\gamma_k[q]}$ and $\Ldp{\xi_k[q]}{\eta_k[q]}$, the global lower bound  of $\Tpp(\bm{\Omega})$ can be found as
\begin{align*}
	\!\!\!\!\!\Tpp(\bm{\Omega})
	&\geq  \bar{\kappa} \abs{\alpha}^2 \sum_{q=1}^Q \sum_{k=1}^{K} \left( \Ldp{\xi_k[q]}{\gamma_k[q]} \left[ \hat{\mD}_{\beta}^{\phi \phi}[q]  \right]_{kk} \right. \\
	&\hspace{1.5cm} \left. + \Ldp{\xi_k[q]}{\eta_k[q]} c_{\phi \phi}[q] \right) \triangleq \Tpp^{(i)}(\bm{\Omega}). \nbthis \label{tau_phi_phi_lwb}
\end{align*}

We further derive an upper bound for the complex term $\Tpap(\bm{\Omega})$. This presents a significant challenge due to the complexity of triples in $\Tpap(\bm{\Omega})$. To address this, we leverage the unique properties of each component within this term. We observe that the real 
component of $\Tpa(\bm{\Omega})$ exhibits a complex nature due to the presence of $\alpha^*$ and the asymmetrical structure of $\hat{\mD}_{\beta}^{\phi}[q]$ and $c_{\phi}[q]$. Consequently, multiplying by $[1,j]$ results in a vector whose first element comprises the real parts of $\re{\alpha^* \hat{\mD}_{\beta}^{\phi}[q]}$ and $\re{\alpha^* c_{\phi}[q]}$, and whose second element contains the negative imaginary parts of these same terms, \textit{i.e.}, $\im{\alpha^* \hat{\mD}_{\beta}^{\phi}[q]}$ and $\im{\alpha^* c_{\phi}[q]}$. From this observation, the vector $\Tpa(\bm{\Omega})$ can be expressed as
\begin{align*}
		\!\!\!&\Tpa(\bm{\Omega}) \!
		=\! \bar{\kappa}  \Big[ \!\sum_{q=1}^Q \bm{\xi}^\T[q] \!\left( \re{\alpha^* \hat{\mD}_{\beta}^{\phi}[q]} \bm{\gamma}[q] \!+\! \bm{\eta}[q] \re{\alpha^* c_{\phi}[q]} \right), \\
		&\hspace{0.6cm}-\sum_{q=1}^Q \bm{\xi}^\T[q] \!\left( \im{\alpha^* \hat{\mD}_{\beta}^{\phi}[q]} \bm{\gamma}[q] \!\!+\! \bm{\eta}[q] \im{\alpha^* c_{\phi}[q]} \right) \Big]. \nbthis \label{tpa_extract}
	\end{align*} 
Here, $\Taa(\bm{\Omega})$ is represented as the sum of a real scalar and the $2 \times 2$ identity matrix $\mI_2$. Its inverse, $\Taa^{-1}(\bm{\Omega})$, is given by:
\begin{align*}
	\!\!\!\Taa^{-1}(\bm{\Omega}) \!=\! \Big[\bar{\kappa} \sum_{q=1}^Q \bm{\xi}^\T[q] \left(  \hat{\mD}_{\beta}[q] \bm{\gamma}[q] + \bm{\eta}[q] c_{2}[q] \right) \Big]^{-1} \mI_2. \nbthis \label{taa_inv}
\end{align*}	
Let $\reim \triangleq \{\resym, \imsym\}$ represent the set containing the real and imaginary components, and define $f_{\reim}(x) \triangleq \reim(x)$ as the operation that applies the corresponding element of $\reim$ to the input $x$. The results in \eqref{tpa_extract} and \eqref{taa_inv} yield a reformulated expression for $\Tpap(\bm{\Omega})$, presented in \eqref{tTt_ori} at the top of the next page. To handle ${\TpapwhX (\bm{\Omega}); \whreim \in \reim}$, we apply the same approximation method as in \eqref{crlb_soc}.
	\begin{figure*}
		\begin{align*}
			\Tpap(\bm{\Omega})
			&\!=\! \bar{\kappa} \sum_{\whreim \in \reim} \frac{ \left[ \sum_{q=1}^Q \bm{\xi}^\T[q] \left( f_{\whreim}\Big({\alpha^* \hat{\mD}_{\beta}^{\phi}[q]}\Big) \bm{\gamma}[q] + \bm{\eta}[q] f_{\whreim}\Big({\alpha^* c_{\phi}[q]}\Big) \right) \right]^2 }{\sum_{q=1}^Q \bm{\xi}^\T[q] \left(  \hat{\mD}_{\beta}[q] \bm{\gamma}[q] + \bm{\eta}[q] c_{2}[q] \right)}  \\
			&\!=\! \bar{\kappa} \sum_{\whreim \in \reim} \underbrace{\frac{  \left[ \sum_{q=1}^Q \sum_{k=1}^{K} \left( \xi_k[q] \gamma_k[q] f_{\whreim}\Big({\alpha^* \left[ \hat{\mD}_{\beta}^{\phi}[q] \right]_{kk}}\Big) + \xi_k[q] \eta_k[q] f_{\whreim}\Big({\alpha^* c_{\phi}[q]}\Big) \right) \right]^2 }{\sum_{q=1}^Q \sum_{k=1}^{K} \left( \xi_k[q] \gamma_k[q] \left[ \hat{\mD}_{\beta}[q] \right]_{kk} + \xi_k[q] \eta_k[q] c_{2}[q] \right)} }_{\triangleq \TpapwhX (\bm{\Omega})}. \nbthis \label{tTt_ori}
		\end{align*}
		\hrule
		\hrule
	\end{figure*}	
	In particular, we introduce the auxiliary variable sets $\{\slvar{I}{x}{\reim}, \slvar{I}{y}{\reim}, \slvar{I}{z}{\reim}\}$ to deal with the complex rational part in $\TpapR (\bm{\Omega})$ and $\TpapI (\bm{\Omega})$, respectively, results in
\begin{align*}
	\TpapX (\bm{\Omega}) \leq \frac{(\slvar{I}{x}{\reim})^2}{\slvar{I}{y}{\reim}} \leq \slvar{I}{z}{\reim}. \nbthis \label{tTt_soc}
\end{align*}
The second constraint in \eqref{crlb_soc_a} is rewritten equivalently as
\begin{align*}
	\Tpp^{(i)}(\bm{\Omega}) - \bar{\kappa} \left(\slvar{I}{z}{\resym} + \slvar{I}{z}{\imsym} \right) \geq \slvar{A}{y}{}. \nbthis \label{crlb_soc_a2_approx}
\end{align*}
We note that the corresponding SOC constraint for \eqref{tTt_soc} is given by
\begin{align*}
	\Big\|\Big[ 2 \slvar{I}{x}{\reim};\; \Big(\slvar{I}{y}{\reim} - \slvar{I}{z}{\reim} \Big)\!\Big]\!\Big\|_2\! \leq \!\Big(\slvar{I}{y}{\reim} + \slvar{I}{z}{\reim} \Big), \nbthis \label{tTt_soc_b} 
\end{align*}
with
\begin{align*}
	\!&\sum_{q=1}^Q \sum_{k=1}^{K} \left( \xi_k[q] \gamma_k[q] f_{\reim} \left(\alpha^* \left[ \hat{\mD}_{\beta}^{\phi}[q] \right]_{kk}\right) \right. \\
	&\hspace{3cm} \left. + \xi_k[q] \eta_k[q] f_{\reim} \left(\alpha^* c_{\phi}[q]\right) \right) \leq \slvar{I}{x}{\reim}, \nbthis \label{tTt_soc_a1}\\
	\!&\sum_{q=1}^Q \sum_{k=1}^{K} \left( \xi_k[q] \gamma_k[q] \left[ \hat{\mD}_{\beta}[q] \right]_{kk} \!\!\!+\! \xi_k[q] \eta_k[q] c_{2}[q] \right) \!\geq\! \slvar{I}{y}{\reim}. \nbthis \label{tTt_soc_a2}
\end{align*} 
It is evident that the denominators in $\TpapR (\bm{\Omega})$ and $\TpapI (\bm{\Omega})$ are identical. Therefore, $\slvar{I}{y}{\reim}$ in \eqref{tTt_soc_a2} can be addressed with a single constraint. Since the left-hand sides (LHS) of \eqref{tTt_soc_a1} and \eqref{tTt_soc_a2} remain non-convex, it is essential to establish bounds to obtain convex forms. Specifically, we reuse the upper bounds $\Udp{\xi_k[q]}{\gamma_k[q]}$ and $\Udp{\xi_k[q]}{\eta_k[q]}$ for \eqref{tTt_soc_a1}, while the lower bounds $\Ldp{\xi_k[q]}{\gamma_k[q]}$ and $\Ldp{\xi_k[q]}{\eta_k[q]}$ are used for \eqref{tTt_soc_a2}. Consequently, the inequalities in \eqref{tTt_soc_a1} and \eqref{tTt_soc_a2} are reformulated as follows:
\begin{align*}
		\!\!\!&\!\!\!\!\sum_{q=1}^Q \sum_{k=1}^{K} \!\!\left( \Udp{\xi_k[q]}{\gamma_k[q]} f_{\reim} \left(\alpha^* \left[ \hat{\mD}_{\beta}^{\phi}[q] \right]_{kk}\right) \right. \\
		&\hspace{2.0cm} \left. + \Udp{\xi_k[q]}{\eta_k[q]} f_{\reim} \left(\alpha^* c_{\phi}[q]\right) \right) \!\leq\! \slvar{I}{x}{\reim}, \nbthis \label{tTt_soc_a1_upb}\\
		\!\!\!&\!\!\!\!\sum_{q=1}^Q \sum_{k=1}^{K} \!\!\left( \Ldp{\xi_k[q]}{\gamma_k[q]} \!\! \left[ \hat{\mD}_{\beta}[q] \right]_{kk} \!\!\!\!\!+\! \Ldp{\xi_k[q]}{\eta_k[q]} c_{2}[q] \right) \!\!\geq\!\! \slvar{I}{y}{\reim}. \nbthis \label{tTt_soc_a2_lwb}
\end{align*} 
\paragraph{Convex Approximation of $\mathtt{CRB}^{-1}_{\phi}(\bm{\Omega})$} The third term in the objective \eqref{eq_obj_dkb}, \textit{i.e.}, $\mathtt{CRB}^{-1}_{\phi}(\bm{\Omega})$ with $\mathtt{CRB}_{\phi}(\bm{\Omega})$ given in \eqref{eq_CRB_phi_result}, can be re-expressed as
\begin{align*}
	\mathtt{CRB}^{-1}_{\phi}(\bm{\Omega}) = \Tpp(\bm{\Omega}) - \Tpap(\bm{\Omega}) - \frac{\Ttp^2(\bm{\Omega})}{\Ttt(\bm{\Omega})}. \nbthis \label{eq_CRB_phi_inv}
\end{align*}
Since $\mathtt{CRB}^{-1}_{\phi}(\bm{\Omega})$ is highly non-convex, we approximate it by deriving a concave lower bound for $\Tpp(\bm{\Omega})$, convex upper bounds for $\Tpap(\bm{\Omega})$ and $\frac{\Ttp^2(\bm{\Omega})}{\Ttt(\bm{\Omega})}$. The derivations of lower bounding $\Tpp(\bm{\Omega})$ and upper bounding $\Tpap(\bm{\Omega})$ are determined as outlined in \eqref{crlb_soc_a2_approx}. To find the upper bound of $\frac{\Ttp^2(\bm{\Omega})}{\Ttt(\bm{\Omega})}$, the auxiliary variables $\{\slvar{B}{x}{}, \slvar{B}{y}{}, \slvar{B}{z}{}\}$ are introduced such that
\begin{align*}
	\frac{\Ttp^2(\bm{\Omega})}{\Ttt(\bm{\Omega})}
	\overset{(a)}{\leq} \frac{\slvar{B}{x}{2}}{\slvar{B}{y}{}} \overset{(b)}{\leq} \slvar{B}{z}{}. \nbthis \label{ttp_ttt_soc}
\end{align*}
Employing the SOC constraint transformation for $(b)$ yields
\begin{align*}
	\Big\|\Big[ 2 \slvar{B}{x}{};\; \Big(\slvar{B}{y}{} - \slvar{B}{z}{} \Big)\!\Big]\!\Big\|_2\! \leq \!\Big(\slvar{B}{y}{} + \slvar{B}{z}{} \Big), \nbthis \label{ttp_ttt_soc_b} 
\end{align*}
with
\begin{align*}
	\!\!\!\!\!\!\sum_{q=1}^Q \sum_{k=1}^{K} \left( \xi_k[q] \gamma_k[q] \left[  \hat{\mD}_{\beta}^{\theta \phi}[q]  \right]_{kk} \!+\! \xi_k[q] \eta_k[q]  c_{\theta \phi}[q] \right) \!&\leq\! \slvar{B}{x}{}, \nbthis \label{ttp_ttt_soc_a1}\\
	\!\!\!\!\!\!\sum_{q=1}^Q \sum_{k=1}^{K} \left( \xi_k[q] \gamma_k[q] \left[ \hat{\mD}_{\beta}^{\theta \theta}[q]  \right]_{kk} \!+\! \xi_k[q] \eta_k[q] c_{\theta \theta}[q] \right) \!&\geq\! \slvar{B}{y}{}. \nbthis \label{ttp_ttt_soc_a2}
\end{align*} 
The LHS of \eqref{ttp_ttt_soc_a1} and \eqref{ttp_ttt_soc_a2} can be approximated using the upper bounds $\Udp{\xi_k[q]}{\gamma_k[q]}$ and $\Udp{\xi_k[q]}{\eta_k[q]}$ for \eqref{ttp_ttt_soc_a1} and the lower bounds $\Ldp{\xi_k[q]}{\gamma_k[q]}$ and $\Ldp{\xi_k[q]}{\eta_k[q]}$ for \eqref{ttp_ttt_soc_a2}. The inequalities \eqref{ttp_ttt_soc_a1} and \eqref{ttp_ttt_soc_a2} are thus reformulated by
\begin{align*}
		&\sum_{q=1}^Q \sum_{k=1}^{K} \left( \Udp{\xi_k[q]}{\gamma_k[q]} \left[ \hat{\mD}_{\beta}^{\theta \phi}[q] \right]_{kk}  \right. \\
		&\hspace{3.0cm} \left. + \Udp{\xi_k[q]}{\eta_k[q]}  c_{\theta \phi}[q] \right) \leq \slvar{B}{x}{}, \nbthis \label{ttp_ttt_soc_a1_upb}\\
		&\sum_{q=1}^Q \sum_{k=1}^{K} \left( \Ldp{\xi_k[q]}{\gamma_k[q]} \left[  \hat{\mD}_{\beta}^{\theta \theta}[q] \right]_{kk}  \right. \\
		&\hspace{3.0cm} \left. + \Ldp{\xi_k[q]}{\eta_k[q]} c_{\theta \theta}[q]  \right) \geq \slvar{B}{y}{}. \nbthis \label{ttp_ttt_soc_a2_lwb}
	\end{align*} 
As a result, $\mathtt{CRB}^{-1}_{\phi}(\bm{\Omega})$ is iteratively replaced by
\begin{align*}
		\!\!\!\!\!\mathtt{CRB}^{-1}_{\phi}(\bm{\Omega}) \!\geq\! \Tpp^{(i)}(\bm{\Omega}) \!-\! \bar{\kappa} (\slvar{I}{z}{\resym} \!+\! \slvar{I}{z}{\imsym}) \!-\!  \slvar{B}{z}{} \!\triangleq\! \Big[\mathtt{CRB}^{-1}_{\phi}\Big]^{(i)}\!\!(\bm{\Omega}). \nbthis \label{eq_CRB_phi_inv_lwb} 
	\end{align*}

\paragraph{Convex Approximation of $P_{\mathtt{tot}}(\bm{\Omega})$}
For the last term $P_{\mathtt{tot}}(\bm{\Omega})$ in the objective function \eqref{eq_obj_dkb}, it is necessary to provide the convex upper bound approximation to $P_{\mathtt{TX}}(\bm{\Omega})$ given in \eqref{eq_Ptx} due to its nonconvexity. By reemploying $\Udp{\xi_k[q]}{\gamma_k[q]}$ and $\Udp{\xi_k[q]}{\eta_k[q]}$, we have
\begin{align*}
		P_{\mathtt{TX}}(\bm{\Omega}) &\!=\! \sum_{q=1}^Q \!\sum_{k=1}^{K}\! \left( \xi_k[q]\! \gamma_k[q] \!\left[ \bar{\mD}_{\beta} \right]_{kk} \!+\! \xi_k[q] \eta_k[q] \right) \\
		&\!\leq\! \sum_{q=1}^Q\! \sum_{k=1}^{K}\! \left( \Udp{\xi_k[q]}{\gamma_k[q]} \!\left[ \bar{\mD}_{\beta} \right]_{kk} \!+\! \Udp{\xi_k[q]}{\eta_k[q]} \right) \\ 
		&\!\triangleq\! P_{\mathtt{TX}}^{(i)}(\bm{\Omega}). \nbthis \label{PTX}
	\end{align*}		
Recall that $\SEc = \sum_{k=1}^K \SEck(\bm{\Omega}) =  \sum_{k=1}^K \sum_{q=1}^Q\! \log_2 \!\left(\! 1\! +\! \frac{\mathcal{R}_{\mathtt{Nu}, k} (\bm{\Omega})[q]}{\mathcal{R}_{\mathtt{De}, k} (\bm{\Omega})[q]}\right)$ is given in \eqref{eq_rate1}, where $\mathcal{R}_{\mathtt{Nu}, k} (\bm{\Omega})[q]$ and $\mathcal{R}_{\mathtt{De}, k} (\bm{\Omega})[q]$ are expressed as in \eqref{eq_ori_RNu} and \eqref{eq_ori_RDe} as the numerator and denominator of $\SEck(\bm{\Omega})$, respectively. Following the inequalities in \cite[eq(75-76)]{sheng2018power}, we obtain the convex upper bound of $\SEck(\bm{\Omega})$ as follows
\begin{align*}
		\SEck(\bm{\Omega}) 
		&\leq  \frac{1}{\ln2} \sum_{q=1}^Q \left[ {\hat{A}_{k}^{(i)}}[q] + 0.5 {\hat{B}_{k}^{(i)}}[q]   \right. \\
		&\hspace{0.2cm} \left. \times \left( {\hat{C}_{k}^{(i)}}[q] \frac{\Big( \mathcal{R}_{\mathtt{Nu}, k} (\bm{\Omega})[q] \Big)^2}{ \mathcal{R}_{\mathtt{De}, k} (\bm{\Omega})[q]}  +  \frac{\Big({{\hat{C}_{k}^{(i)}}[q]}\Big)^{-1}}{\mathcal{R}_{\mathtt{De}, k} (\bm{\Omega})[q]} \right) \right] \\
		&\triangleq	\SEck^{(i)}(\bm{\Omega}), \nbthis \label{eq_approx_R}
	\end{align*}
where
\begin{align*}
	{\hat{A}_{k}^{(i)}}[q] &\triangleq \ln\Big(1+\frac{\mathcal{R}_{\mathtt{Nu}, k} (\bm{\Omega}^{(i)})[q]}{\mathcal{R}_{\mathtt{De}, k} (\bm{\Omega}^{(i)})[q]}\Big) \\
	& \hspace{2cm} - \frac{\mathcal{R}_{\mathtt{Nu}, k} (\bm{\Omega}^{(i)})[q]}{\mathcal{R}_{\mathtt{Nu}, k} (\bm{\Omega}^{(i)})[q]+\mathcal{R}_{\mathtt{De}, k} (\bm{\Omega}^{(i)})[q]} , \\
	{\hat{B}_{k}^{(i)}}[q] &\triangleq \frac{\mathcal{R}_{\mathtt{De}, k} (\bm{\Omega}^{(i)})[q]}{\mathcal{R}_{\mathtt{Nu}, k} (\bm{\Omega}^{(i)})[q]+\mathcal{R}_{\mathtt{De}, k} (\bm{\Omega}^{(i)})[q]}, \\
	{\hat{C}_{k}^{(i)}}[q] &\triangleq \frac{1}{\mathcal{R}_{\mathtt{Nu}, k} (\bm{\Omega}^{(i)})[q]}.
\end{align*}
Since $\mathcal{R}_{\mathtt{Nu}, k} (\bm{\Omega})[q]$ and $\mathcal{R}_{\mathtt{De}, k} (\bm{\Omega})[q]$ are nonconvex, leading to the non-convexity of $\Big( \mathcal{R}_{\mathtt{\mathtt{Nu}}, k} (\bm{\Omega})[q] \Big)^2/\mathcal{R}_{\mathtt{De}, k} (\bm{\Omega})[q]$. To find the global convex upper bound of $\Big( \mathcal{R}_{\mathtt{Nu}, k} (\bm{\Omega})[q] \Big)^2/\mathcal{R}_{\mathtt{De}, k} (\bm{\Omega})[q]$, we introduce the auxiliary variables $\{\slvar{C}{k,x}{}[q], \slvar{C}{k,y}{}[q], \slvar{C}{k,z}{}[q]\}$ such that
\begin{align*}
	\frac{\Big( \mathcal{R}_{\mathtt{Nu}, k} (\bm{\Omega})[q] \Big)^2}{ \mathcal{R}_{\mathtt{De}, k} (\bm{\Omega})[q]}
	\overset{(a)}{\leq} \frac{\slvar{C}{k,x}{2}[q]}{\slvar{C}{k,y}{}[q]} \overset{(b)}{\leq} \slvar{C}{k,z}{}[q], \forall k,q. \nbthis \label{PR_soc}
\end{align*}
Employing SOC constraint transformation for $(b)$, we obtain
\begin{align*}
	\!\!\!\!\! \Big\|\Big[ 2 \slvar{C}{k,x}{}[q];\; \Big(\slvar{C}{k,y}{}[q] - \slvar{C}{k,z}{}[q] \Big)\!\Big]\!\Big\|_2\! \leq \!\Big(\slvar{C}{k,y}{}[q] + \slvar{C}{k,z}{}[q] \Big), \nbthis \label{PR_soc_b} 
\end{align*}
with
\begin{align*}
	\!\!\!\!\! \mathcal{R}_{\mathtt{Nu}, k} (\bm{\Omega})[q] \!\leq \slvar{C}{k,x}{}[q], \text{and }
	\mathcal{R}_{\mathtt{De}, k} (\bm{\Omega})[q] \!\geq \slvar{C}{k,y}{}[q], \forall k,q. \!\nbthis \label{PR_soc_a}
\end{align*} 

Next, we need to derive a convex upper bound for $\mathcal{R}_{\mathtt{Nu}, k} (\bm{\Omega})[q]$ and concave lower bounds for $\mathcal{R}_{\mathtt{De}, k} (\bm{\Omega})[q]$.
Given the convex upper bounds for $\Udp{\xi_k[q]}{\gamma_k[q]}$, $\Udp{\xi_k[q]}{\eta_k[q]}$, and $\Utp{\xi_k[q]}{\sqrt{\gamma_k[q]}}{\sqrt{\eta_k[q]}}$, we have
	\begin{align*}
		&\mathcal{R}_{\mathtt{Nu}, k} \!(\bm{\Omega})[q] \\ 
		&\!\leq\! {\alpha_{\mathtt{ZF}}^2[q]} \!\Udp{\xi_k[q]}{\gamma_k[q]} \!\!+\!\! {\frac{1}{\Nt}}\!\! \abs{\vh_k^\H [q] \va[q]}^2 \!\Udp{\xi_k[q]}{\eta_k[q]}  \\
		&\hspace{0.2cm} \!+\! 2  {\frac{ {\alpha_{\mathtt{ZF}}[q]}}{\sqrt{\Nt}}} \Utp{\xi_k[q]}{\sqrt{\gamma_k[q]}}{\sqrt{\eta_k[q]}}  \re{\vh_k^\T [q] \va^*[q]}, \\
		&\triangleq \widehat{\mathcal{R}}_{\mathtt{Nu}, k} (\bm{\Omega})[q], \nbthis \label{RNu_upb}
	\end{align*}
which serves as a convex upper bound of $\mathcal{R}_{\mathtt{Nu}, k} (\bm{\Omega})[q]$.  The concave lower bound of $\mathcal{R}_{\mathtt{De}, k} (\bm{\Omega})[q]$ is achieved by determining $\Ldp{\xi_j[q]}{\eta_j[q]}$, which is derived using a method similar to that of $\Ldp{\xi_k[q]}{\eta_k[q]}$, such as:
\begin{align*}
	\mathcal{R}_{\mathtt{De}, k} (\bm{\Omega})[q] &\geq {\frac{1}{\Nt}}  \sum_{j \neq k} \Ldp{\xi_j[q]}{\eta_j[q]} \abs{\vh_k^\H [q] \va[q]}^2  + \sigmac \\
	&\triangleq \widehat{\mathcal{R}}_{\mathtt{De}, k} (\bm{\Omega})[q]. \nbthis \label{RDe_lwb}
\end{align*}
The inequalities \eqref{PR_soc_a} can be rewritten equivalently as:
\begin{align*}
	\!\!\!\!\!\! \widehat{\mathcal{R}}_{\mathtt{Nu}, k} (\bm{\Omega})[q] \!\leq \slvar{C}{k,x}{}[q], \text{and } 
	\widehat{\mathcal{R}}_{\mathtt{De}, k} (\bm{\Omega})[q] \!\geq \slvar{C}{k,y}{}[q], \forall k,q. \!\nbthis \label{PR_soc_a_approx}
\end{align*} 
Substituting \eqref{PR_soc}, \eqref{RNu_upb}, and \eqref{RDe_lwb} into \eqref{eq_approx_R}, we obtain the upper bound of $\SEc$ as follows:
\begin{align*}
		\!\!\!\!\!&\SEck(\bm{\Omega}) \leq  \SEck^{(i)}(\bm{\Omega}) \\
		\!\!\!\!\!&\leq \frac{1}{\ln2} \sum_{q=1}^Q \left[ {\hat{A}_{k}^{(i)}}[q] + 0.5 {\hat{B}_{k}^{(i)}}[q]   \right. \\
		\!\!\!\!\!&\hspace{0.5cm} \left. \times \left( {\hat{C}_{k}^{(i)}}[q] \slvar{C}{k,z}{}[q]  \!+\! \frac{\Big({{\hat{C}_{k}^{(i)}}[q]}\Big)^{-1}}{\widehat{\mathcal{R}}_{\mathtt{De}, k} (\bm{\Omega})[q]} \right) \right] \!\triangleq\!	\whSEck^{(i)}(\bm{\Omega}). \nbthis \label{eq_approx_R_new}
\end{align*}

From \eqref{PTX} and \eqref{eq_approx_R_new}, we have
\begin{align*}
    \!\!\!\!\! P_{\mathtt{tot}}(\bm{\Omega}) \!\leq\! \frac{1}{\varrho} P_{\mathtt{TX}}^{(i)}(\bm{\Omega}) \!+\! P_0 \!+\! \epsilon \sum_{k=1}^{K} \whSEck^{(i)}(\bm{\Omega}) \!\triangleq\! P^{(i)}(\bm{\Omega}). \nbthis \label{eq_total_power_approx}
\end{align*}	

\subsubsection{Convexifying Rate Threshold Constraint} The convexity of the constraint \eqref{cons_comm_OSA} is achieved by deriving a concave lower bound for $\SEck(\bm{\Omega})$, as formulated in \eqref{eq_approx_Rk_new}, which can be expressed as:
\begin{align*}
	\wtSEck^{(i)}(\bm{\Omega}) \geq \SEth, \forall k. \nbthis \label{cons_comm_OSA_approx}
\end{align*}

\subsubsection{Convexifying CRB Constraints \eqref{eq::probOSA::cons_CRB_theta} and \eqref{eq::probOSA::cons_CRB_phi}} We can rewrite \eqref{eq::probOSA::cons_CRB_theta} and \eqref{eq::probOSA::cons_CRB_phi} equivalently as
\begin{align*}
	\mathtt{CRB}^{-1}_{\theta}(\bm{\Omega})
	\geq \frac{1}{\mathtt{CRB}_{\theta}^0}, \text{and }
		\mathtt{CRB}^{-1}_{\phi}(\bm{\Omega})
	\geq \frac{1}{\mathtt{CRB}_{\phi}^0}. \nbthis \label{eq_SOC}
\end{align*}
The LHS of constraints \eqref{eq_SOC} is derived as in \eqref{eq_CRB_theta_inv_lwb} and \eqref{eq_CRB_phi_inv_lwb}, respectively. The convex form of CRB constraints \eqref{eq::probOSA::cons_CRB_theta} and \eqref{eq::probOSA::cons_CRB_phi} are rewritten as
\begin{align*}
	\Big[\mathtt{CRB}^{-1}_{\theta}\Big]^{(i)}(\bm{\Omega})
	\geq \frac{1}{\mathtt{CRB}_{\theta}^0}, \text{and } \Big[\mathtt{CRB}^{-1}_{\phi}\Big]^{(i)}(\bm{\Omega})
	\geq \frac{1}{\mathtt{CRB}_{\phi}^0}. \nbthis \label{eq_SOC_approx}
\end{align*} 	

\subsubsection{Convexifying Power Constraint \eqref{eq::probOSA::cons_power_MRT}} Since The LHS of the power budget constraint \eqref{eq::probOSA::cons_power_MRT} is nonconvex; however, its convexity can be achieved by applying the upper bound approximation previously analyzed in \eqref{PTX}. Thus, we obtain:
\begin{align*}
	P_{\mathtt{TX}}^{(i)}(\bm{\Omega}) \leq P_{\mathtt{max}}. \nbthis \label{PTX_approx}
\end{align*}

\subsection{Proposed Iterative Algorithm}
Define by $\bm{\Theta} \triangleq \{\bm{\Theta}_{\mathcal{A}},\bm{\Theta}_{\mathcal{I}},\bm{\Theta}_{\mathcal{B}},\bm{\Theta}_{\mathcal{C}} \}$ the set of auxiliary variables used in the approximation process, where $\bm{\Theta}_{\mathcal{A}} = \{ \{\mathcal{A}_{i} \}, i=\{x,y,z\} \}$, $\bm{\Theta}_{\mathcal{I}} = \{ \{\mathcal{I}_{i}^{\reim} \}, \reim = \{\resym,\imsym\}, i=\{x,y,z\} \}$, $\bm{\Theta}_{\mathcal{B}} = \{ \{\mathcal{B}_{i} \}, i=\{x,y,z\} \}$, $\bm{\Theta}_{\mathcal{C}} = \{ \{\mathcal{C}_{k,i}[q] \}, \forall k,\forall q, i=\{x,y,z\} \}$. The approximate convex program solved at iteration $i$ for \eqref{ori_probOSA}  is
\begin{subequations}\label{ori_probOSA_approx}
	\begin{IEEEeqnarray} {rcl}
			&& \underset{\substack{ \bm{\Omega}, \bm{\Theta}}}{\textrm{maximize }}   \;  \Big\{ \sum_{k=1}^{K} \wtSEck^{(i)}(\bm{\Omega}) + \omega \Big( \Big[\mathtt{CRB}^{-1}_{\theta}\Big]^{(i)}(\bm{\Omega})  \nonumber \\
			&&\qquad\qquad\quad + \Big[\mathtt{CRB}^{-1}_{\phi}\Big]^{(i)}(\bm{\Omega}) \Big)  - \tau^{(i)} P^{(i)}(\bm{\Omega}) \Big\} \label{ori_probOSA_approx_obj} \\
			&&\mathtt{s.t.}\ 
			 \eqref{eq_total_power_0}, \eqref{eq_xi_power}, \eqref{crlb_soc_b}, \eqref{crlb_soc_a1_approx}, \eqref{crlb_soc_a2_approx}, \eqref{tTt_soc_b}, \eqref{tTt_soc_a1_upb}, \eqref{tTt_soc_a2_lwb}, \nonumber \\ 
 &&\qquad \eqref{ttp_ttt_soc_b}, \eqref{ttp_ttt_soc_a1_upb}, \eqref{ttp_ttt_soc_a2_lwb},
			 \eqref{PR_soc_b}, \eqref{PR_soc_a_approx}, \eqref{cons_comm_OSA_approx}, \eqref{eq_SOC_approx}, \eqref{PTX_approx}.
	\end{IEEEeqnarray}
\end{subequations}
At iteration ${i}$,  $\tau^{(i)}$ is updated as
\begin{align*}
	\!\!\! \tau^{(i)}\!\! = \!\!\frac{\sum_{k=1}^{K}\!\! \wtSEck^{(i)}(\bm{\Omega}) \!+\! \omega \Big( \Big[\mathtt{CRB}^{-1}_{\theta}\Big]^{(i)}\!\!\!(\bm{\Omega})\!\! +\!\! \Big[\mathtt{CRB}^{-1}_{\phi}\Big]^{(i)}\!\!\!(\bm{\Omega}) \Big)}{P^{(i)}(\bm{\Omega})}  \nbthis. \label{tau_approx}
\end{align*}

\subsubsection{Initialization} To find an initial feasible point for problem (\ref{ori_probOSA_approx}), we begin with any point that satisfies the convex constraints (\ref{eq_total_power_0}) and (\ref{eq_xi_power}), and then iterate
\begin{subequations}\label{ori_probOSA_initial}
	\begin{IEEEeqnarray} {rcl}
			&&\underset{\substack{ \bm{\Omega}, \bm{\Theta}}}{\mathrm{maximize}}  \quad \min \Bigg( \Bigg\{ \min_{k=1,\cdots, K} \Bigg[ \frac{\wtSEck^{(i)}(\bm{\Omega})}{\SEth} -1 \Bigg]\Bigg\},  \nonumber\\ 
			&&\qquad  \Bigg\{ \frac{\Big[\mathtt{CRB}^{-1}_{\theta}\Big]^{(i)}(\bm{\Omega})}{1/\mathtt{CRB}_{\theta}^0} -1\Bigg\},   \Bigg\{ \frac{\Big[\mathtt{CRB}^{-1}_{\phi}\Big]^{(i)}(\bm{\Omega})}{1/\mathtt{CRB}_{\phi}^0} -1\Bigg\} \Bigg) \label{ori_probOSA_initial_obj} \hspace{1cm} \\
			&&\mathtt{s.t.}\quad \eqref{eq_total_power_0}, \eqref{eq_xi_power}, \eqref{crlb_soc_b}, \eqref{crlb_soc_a1_approx}, \eqref{crlb_soc_a2_approx}, \eqref{tTt_soc_b},\nonumber \\
			& & \qquad\ \; \eqref{tTt_soc_a1_upb}, \eqref{tTt_soc_a2_lwb}, \eqref{ttp_ttt_soc_b},  \eqref{ttp_ttt_soc_a1_upb}, \eqref{ttp_ttt_soc_a2_lwb}, \eqref{PTX_approx}, 
	\end{IEEEeqnarray}
\end{subequations}
until the objective value in (\ref{ori_probOSA_initial_obj}) reaches or surpasses zero. 


\subsubsection{Complexity and Convergence Analysis} The computational complexity of Algorithm \ref{alg1} is ${\cal O}(n^2 m^{2.5} + m^{3.5})$, where $n = 6KQ + 12$ represents the number of scalar variables, comprising $3KQ$ from $\bm{\Omega}$ and $12 + 3KQ$ from $\bm{\Theta}$. The parameter $m = 15 + 7KQ + K$ denotes the number of corresponding constraints~\cite[p.4]{peaucelle2002user}.
	
	We denote $f(\bm{\Omega}) = \SEc + \omega( \mathtt{CRB}^{-1}_{\theta} (\bm{\Omega}) + \mathtt{CRB}^{-1}_{\phi} (\bm{\Omega}) )  - \tau P_{\mathtt{tot}} (\bm{\Omega})$ as the objective function defined in \eqref{eq_obj_dkb}. Given a feasible set $\bm{\Omega}^{(i)}$ for (\ref{ori_probOSA_approx}), our approximation method addresses the non-concavity of  $f(\bm{\Omega})$ by constructing a concave function $f^{(i)} (\bm{\Omega}) = \sum_{k=1}^{K} \wtSEck^{(i)} (\bm{\Omega}) + \omega \Big( \Big[\mathtt{CRB}^{-1}_{\theta}\Big]^{(i)} (\bm{\Omega}) 
	\Big[\mathtt{CRB}^{-1}_{\phi}\Big]^{(i)} (\bm{\Omega}) \Big)  - \tau^{(i)} P^{(i)} (\bm{\Omega})$. This concave approximation ensures that the objective function in (\ref{ori_probOSA_approx}) is lower bounded by
	\begin{align*}
		f(\bm{\Omega}) \leq f^{(i)}(\bm{\Omega}). \nbthis \label{proof_convex}
	\end{align*}
	It coincides with $f(\bm{\Omega})$ within the feasible region  $\bm{\Omega}^{(i)}$
	\begin{align*}
		f(\bm{\Omega}^{(i)}) = f^{(i)}(\bm{\Omega}^{(i)}). \nbthis \label{proof_convex1}
	\end{align*}
	The above analysis shows that 
	\begin{align*}
		f(\bm{\Omega}^{(i)})
		&\overset{(a)}{=} f^{(i)}(\bm{\Omega}^{(i)}) 
		&\overset{(b)}{\leq} f^{(i)} (\bm{\Omega}^{(i+1)}) 
		&\overset{(c)}{\leq} f^{(i+1)} (\bm{\Omega}^{(i+1)})
	\end{align*}
	for every $ \bm{\Omega}^{(i+1)} \neq  \bm{\Omega}^{(i)}$,
	where $(a)$ follows from \eqref{proof_convex1}, $(b)$ holds because $\bm{\Omega}^{(i+1)}$ and $\bm{\Omega}^{(i)}$ are
	the optimal solution and a feasible point for (\ref{ori_probOSA_approx}), respectively, and $(c)$ follows from \eqref{proof_convex}. Algorithm \ref{alg1} generates an improving sequence of feasible points $(\bm{\Omega}^{(i)})$ for problem (\ref{ori_probOSA_approx}). Following arguments similar to those in \cite{kha2011fast}, it can be demonstrated that Algorithm \ref{alg1} converges to at least a locally optimal solution of (\ref{ori_probOSA}) satisfying the KKT optimality conditions. The complete SCA-based solution procedure is outlined in Algorithm \ref{alg1}.

\begin{algorithm}[t]
	\begin{algorithmic}[1]
		\fontsize{10}{10}\selectfont
		\protect\caption{Proposed Iterative Algorithm for Solving  \eqref{ori_probOSA}}
		\label{alg1}
		\global\long\def\algorithmicrequire{\textbf{Initialization:}}
		\REQUIRE  Set $i:=1$ and randomly generate an initial point satisfying \eqref{eq_total_power_0} and \eqref{eq_xi_power}. Solve \eqref{ori_probOSA_initial} to search the feasible set $\bm{\Omega}^{(i)}$.
		\REPEAT
		\STATE Solve  \eqref{ori_probOSA_approx} to obtain the optimal variable set of $\bm{\Omega}^{\star}$;
		\STATE Update:\ \ $\bm{\Omega}^{(i)} := \bm{\Omega}^{\star}$;
		\STATE Update:\ \ $\tau^{(i)} $ according to \eqref{tau_approx};
		\STATE Set $i:=i+1$;
		\UNTIL Convergence\\
		\STATE{\textbf{Output:}} $\bm{\Omega}^{\star}$.
	\end{algorithmic} 
\end{algorithm}

\section{Numerical Results}
\label{sec_sim}
\subsection{Simulation Setup}
In this section,  we present and analyze numerical results to verify the effectiveness of the proposed method. We consider a scenario in which users' locations are randomly distributed according to a uniform distribution within a circular cell of radius $1000$ m. The BS is positioned at the center, $(0,0)$, with a minimum distance constraint of $r_{\mathtt{h}} = 100$ m from any user. The target is assumed to be located $400$ m from the BS at angles $(\theta,\phi) = \left(\frac{\pi}{8}, \frac{\pi}{4}\right)$.
The large-scale fading parameters are modeled as $\beta_k = z_k/(r_k/r_{\mathtt{h}})^{\nu}$, where $z_k$ follows a log-normal distribution with a standard deviation of $\sigma_{\mathtt{shadow}}$, and $r_k$ and $\nu$ represent the distance between UE $k$ and the BS, and the path loss exponent, respectively ~\cite{ngo2013energy}.
The system bandwidth, denoted by \text{BW}, is divided into multiple subcarriers, each with a bandwidth of $\text{BW}/Q$. We set $L=30$ and $\sigmac = \sigmas = \noise = 1$. For the array response vectors, $\ah[q]$ and $\av[q]$, we define the number of elements as $\Nth = \Ntv = \sqrt{\Nt} \in \mathbb{N}$ and $\Nrh = \Nrv = \sqrt{\Nr} \in \mathbb{N}$.
For simplicity, we consider $\mathtt{CRB}^0=\mathtt{CRB}_{\theta}^0 = \mathtt{CRB}_{\phi}^0$. Additional simulation parameters are provided in Table \ref{table:1}.
\begin{table}[t]
	\centering
	\caption{Parameter Settings}
	\begin{tabular}{|p{5.2cm}|p{2.5cm}|}
		\hline
		Parameter & Value \\
		\hline
            \hline
		Center frequency $f_c$ &  2GHz \\
        System bandwidth \text{BW}  &  10MHz  \\
		Number of communication UEs $K$ & 6 \\
		System circuit power $P_0$   &  5.6 mW  \\
		Power consumption coefficient $\epsilon$  &  -26dBm/bps \cite{xiong2011energy}   \\
		BS amplifier efficiency $\varrho$	& 0.35	\\
		Shadowing standard deviation & 7 dB \\
		Path loss exponent $\nu$ & 3.2 \\
		\hline
	\end{tabular}
	\label{table:1}
\end{table}

We employ CVX to solve the convex program \eqref{ori_probOSA_approx}. For the performance comparison, we consider the two benchmark approaches as follows:
\begin{itemize}
		\item \textbf{Equal power fractions among communication users (\textit{EqualCom})}: This scheme allocates power between communication and sensing while distributing it equally among all communication users. It is derived by solving problem \eqref{ori_probOSA_approx} under the constraint $\gamma_1[q] = \gamma_2[q] = \cdots = \gamma_K[q], \forall q$. 
		Compared to Algorithm 1, the \textit{EqualCom} scheme incurs lower computational complexity. This reduction stems from the fewer scalar variables involved in $\bm{\Omega}$, specifically $2KQ + Q$, which leads to $n=5KQ + Q + 12$ variables, and $m=15 + 6KQ + K + Q$ constraints \cite{peaucelle2002user}.
		\item \textbf{Equal power allocation between communication and sensing as well as among communication users (\textit{EqualC\&S})}:  In this setup, power is evenly divided between communication and sensing for each user on every subcarrier, ensuring $\gamma_k[q] = \eta_k[q], \forall k$. This approach adheres to the conditions in \eqref{eq::probOSA::cons_power_MRT} and \eqref{eq_total_power_0}, distributing the power budget uniformly for joint transmission.
\end{itemize}

For brevity, we denote the sum EE objective function in \ref{eq::probOSA::obj_func_OSA} as  ``Overall EE". Similarly, ``Com. EE" and ``Sen. EE" represent  $\mathtt{EE}_{\mathtt{c}}(\bm{\Omega})$ and $\omega \mathtt{EE}_{\mathtt{s}}(\bm{\Omega})$, respectively.

\subsection{Tradeoff between Communication and Sensing Performance }

Fig.~\ref{fig_omega} demonstrates the impact of parameter $\omega$ on balancing communication and sensing EE. Since $\mathtt{EE}_{\mathtt{s}}(\bm{\Omega})$ is in different scale compare to ``Overall EE", ``Com. EE" and ``Sen. EE", it is challenging to individually demonstrate the impact of $\omega$ on $\mathtt{EE}_{\mathtt{s}}(\bm{\Omega})$. For easier comparison, we normalize $\mathtt{EE}_{\mathtt{s}}(\bm{\Omega})$ to its maximum value, defining it as ``Nor. EEs'', which ranges within $[0, 1]$. As shown in Fig.~\ref{fig_omega}, increasing $\omega$ enhances the sensing EE while reducing communication EE, emphasizing sensing functionality. Since ``Nor. EEs'' is independent of $\omega$, the plot confirms that the trend in the overall EE, dominated by communication EE at low $\omega$ and by sensing EE at high $\omega$, also holds for sensing EE. As shown in Fig.~\ref{fig_omega}, a closer analysis of the low and high $\omega$ regimes is essential to better understand the trade-off between communication and sensing performance. Variations in $\omega$ can substantially shift the priority between these functions. To represent these regimes in the subsequent numerical analysis, we set $\omega$ to $10^{-4}$ for the low $\omega$ regime (LoR) and $2 \times 10^{-3}$ for the high $\omega$ regime (HoR).
\begin{figure}[t]%
	\centering
	\includegraphics[scale=0.55]{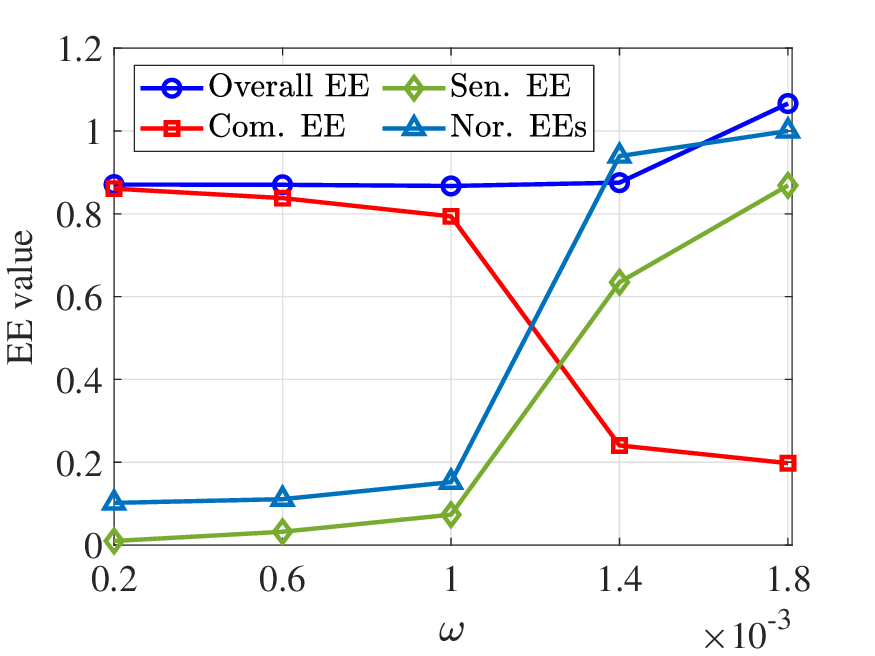}
	\caption[]{EE versus $\omega$ with $\Nt=25$, $\Nr=25$, $Q=16$, $\mathtt{CRB}^0 = -30$ dB, $\SEth = 5.0$ bps/Hz, and $P_{\mathtt{max}}=20$ dBm .}%
	\label{fig_omega}
\end{figure}

\subsection{Convergence of Algorithm \ref{alg1}}
\begin{figure}[t]%
	\centering
	\includegraphics[scale=0.55]{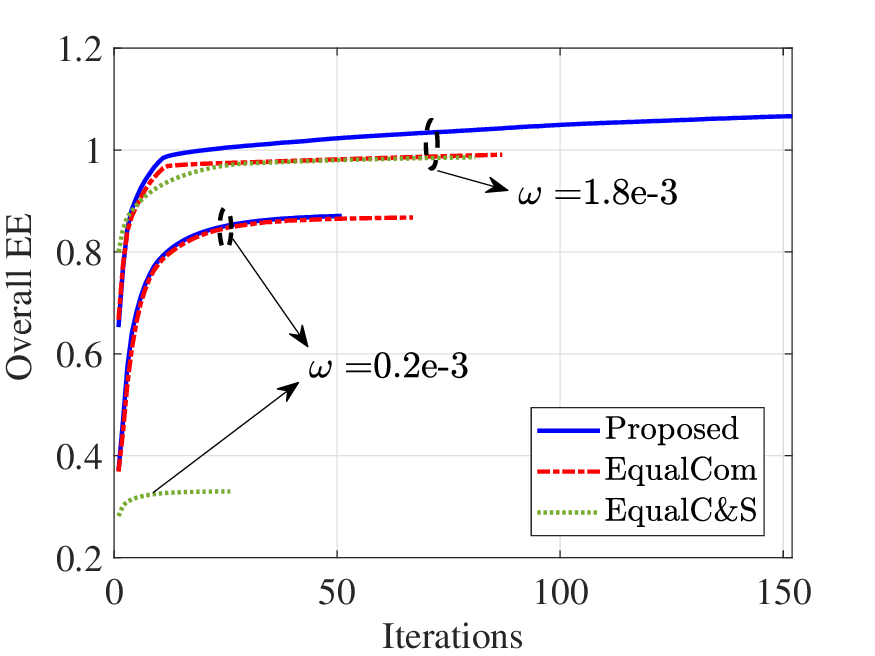}
	\caption[]{Convergence of Algorithm \ref{alg1} with $\Nt=25$, $\Nr=25$, $Q=16$, $\mathtt{CRB}^0 = -30$ dB, $\SEth=5$ bps/Hz, and $P_{\mathtt{max}}=20$ dBm.} 
	\label{fig_conv}
\end{figure}

To verify Algorithm \ref{alg1}'s convergence, we utilize parameters from  Fig. \ref{fig_omega} with $\omega = \{0.2 \times 10^{-3},1.8\times 10^{-3}\}$ and a convergence tolerance of $10^{-4}$. 
Fig.\ \ref{fig_conv} illustrates the convergence behavior of our proposed algorithm and two baselines, showing the evolution of overall EE over iterations. All schemes exhibit distinct convergence characteristics, consistently improving performance with each iteration until convergence is reached. 
At $\omega = 0.2 \times 10^{-3}$, all schemes converge within approximately 50 iterations. However, at higher $\omega$, convergence requires more iterations. Specifically, our proposed scheme takes 150 iterations to converge, while the baselines converge in around 80 iterations.
Each iteration takes approximately 45 seconds on a Core i5 machine (6 cores, 2.7 GHz), utilizing 25\% of the CPUs and 1.5 GB of RAM for running MATLAB with the CVX solver to find the optimal solution.
Overall EE of all schemes performs worse at $\omega = 0.2 \times 10^{-3}$ than at higher $\omega$, \textit{i.e.}, $1.8\times 10^{-3}$. Our proposed algorithm exhibits similar overall EE trends to \textit{EqualCom} at low $\omega$. However, at high $\omega$, it demonstrates a clear performance improvement.

\subsection{Communication and Sensing Performance }

\begin{figure}[t]%
	\vspace{-0.5cm}\centering
	\includegraphics[scale=0.55]{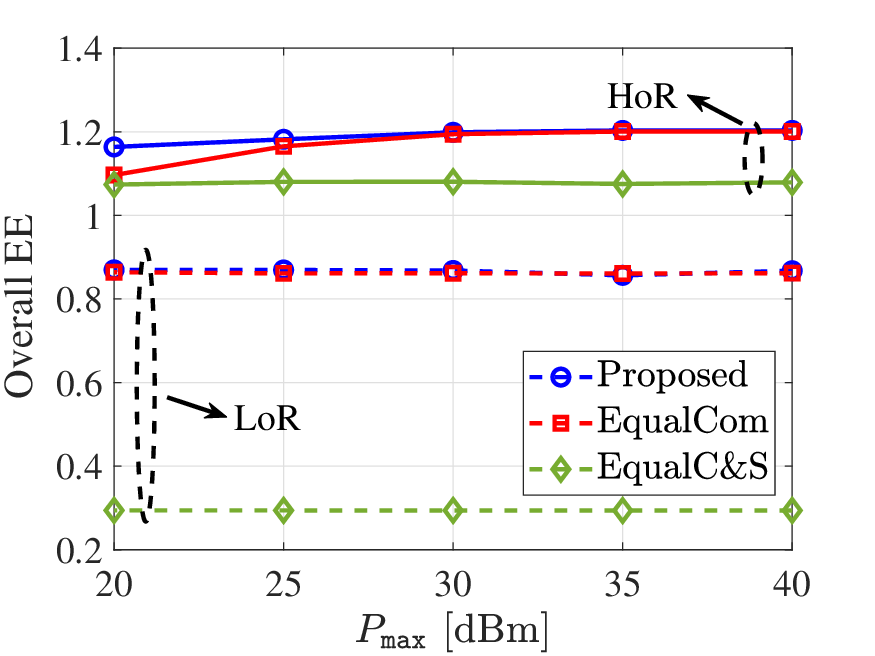}
	\caption[]{The overall EE versus $P_{\mathtt{max}}$ with $\Nt=25$, $\Nr=25$, $Q=16$, $\mathtt{CRB}^0=-30$ dB, and $\SEth=5.0$ bps/Hz.}%
	\label{fig_EE_Pt}
\end{figure}

Fig. \ref{fig_EE_Pt} plots the overall EE versus $P_{\mathtt{max}}$ for $\Nt=25$, $\Nr=25$, $Q=16$, $\mathtt{CRB}^0=-30$ dB, and $\SEth=5.0$ bps/Hz. Within the transmit power range of $20$ to $40$ dBm, all schemes exhibit consistent performance in the LoR. However, in the HoR, the proposed and \textit{EqualCom} schemes show noticeable performance improvements between 20 and 30 dBm, stabilizing thereafter.
Notably, the proposed algorithm significantly outperforms \textit{EqualCom} at low transmit power levels, particularly around 20 dBm. This improvement is primarily due to the influence of the sensing energy efficiency (Sen. EE) on the overall EE in the HoR. As $P_{\mathtt{max}}$ increases, the sensing function gains greater flexibility to allocate power for meeting system thresholds, thereby enhancing the Sen. EE. For the remaining cases, the \textit{EqualCom} scheme achieves comparable performance to the proposed scheme, but with lower computational complexity.

\begin{figure*}[t!]
	\centering
	\subfigure[Comparison to baselines with $\SEth=5.0$ bps/Hz and $P_{\mathtt{max}}=30$ dBm.]{%
		\label{fig_EE_CRLB}%
		\includegraphics[width=0.33\textwidth]{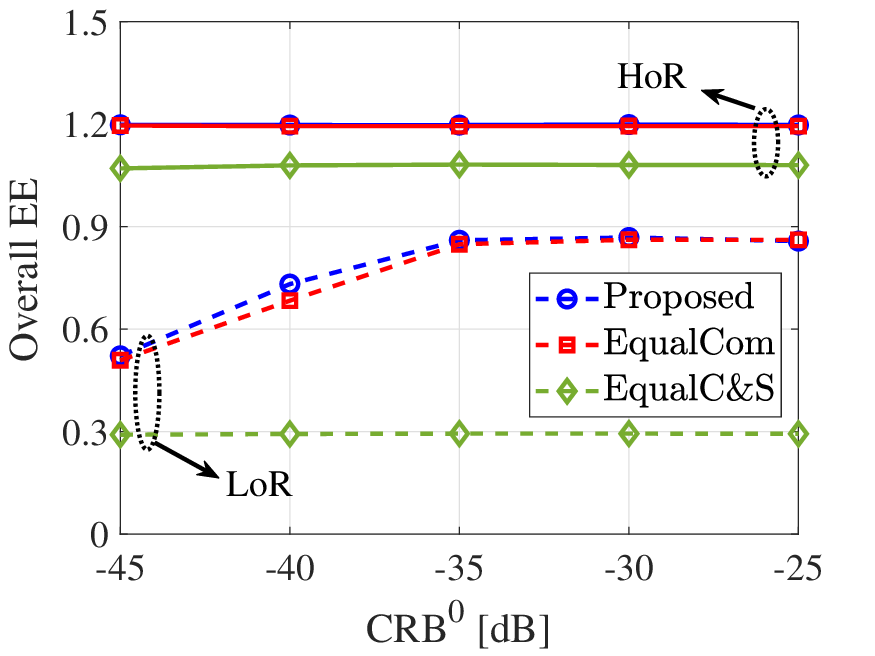}}%
	~
	\subfigure[Communication and sensing EE with $\SEth=5.0$ bps/Hz.]{%
		\label{fig_EE_CRLB_Pt2030}%
		\includegraphics[width=0.34\textwidth]{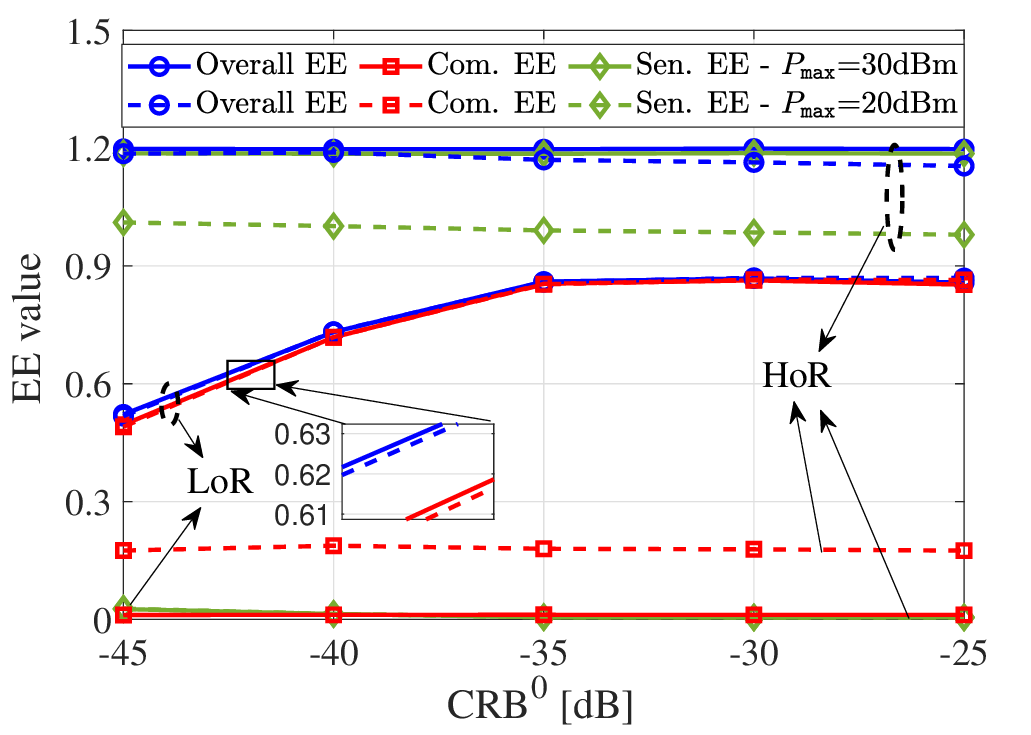}}%
	~
	\subfigure[Different $\SEth$ level plot with $P_{\mathtt{max}}=20$ dBm at LoR.]{%
		\label{fig_EE_CRB_lw}%
		\includegraphics[width=0.33\textwidth]{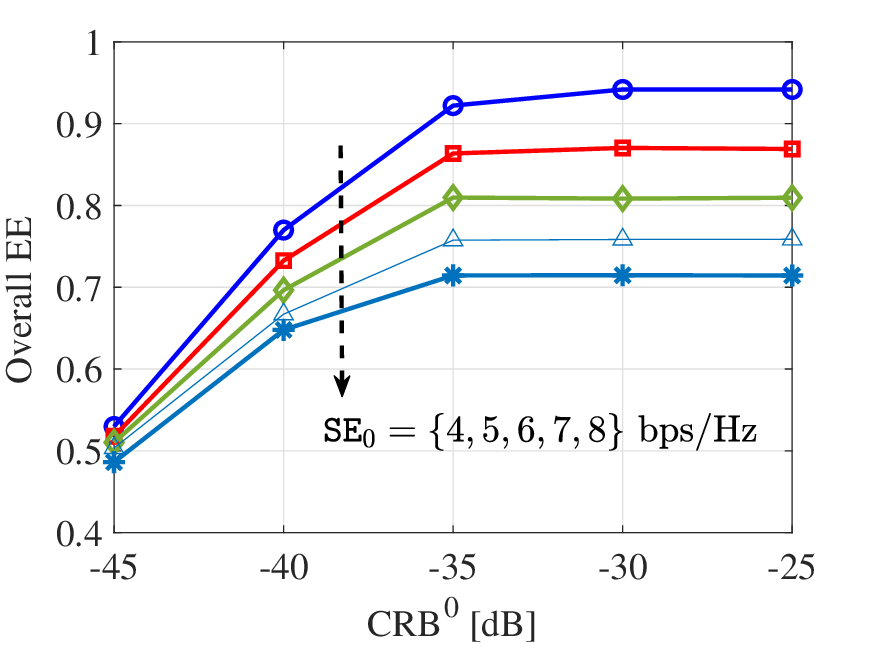}}%
	\caption{EE versus $\mathtt{CRB}^0$ with $\Nt=25$, $\Nr=25$, and $Q=16$.}
	\label{fig_CRLB}
\end{figure*}

Fig. \ref{fig_CRLB} shows the overall EE as a function of $\mathtt{CRB}^0=\mathtt{CRB}_{\theta}^0=\mathtt{CRB}_{\phi}^0$ for $\SEth=5.0$ bps/Hz and $P_{\mathtt{max}}=30$ dBm. As depicted in Fig. \ref{fig_EE_CRLB}, \textit{EqualCom} emerges as a strong competitor, whereas \textit{EqualC\&S} underperforms. The proposed scheme surpasses \textit{EqualCom} in the range $-45 \text{ dB} \leq \mathtt{CRB}^0 \leq -35 \text{ dB}$ and converges to its maximum at higher $\mathtt{CRB}^0$ values. Although higher $\mathtt{CRB}^0$ demands greater power consumption, the proposed scheme benefits from enhanced rate adaptation flexibility, improving the communication EE (EEc). Conversely, the uniform power allocation in \textit{EqualCom} restricts spectral efficiency (SE), imposing stricter EEc requirements in the LoR. Fig. \ref{fig_EE_CRLB_Pt2030} further illustrates that increasing $\mathtt{CRB}^0$ raises both communication EE and overall EE in the LoR. However, in the HoR, overall EE remains largely unaffected for $P_{\mathtt{max}}=30$ dBm, while a slight decrease is observed at $P_{\mathtt{max}}=20$ dBm. Lower $\mathtt{CRB}^0$ leads to reduced transmit power requirements, enhancing communication EE and overall EE in the LoR. Conversely, higher $\mathtt{CRB}^0$ degrades sensing EE in the HoR, as theoretically supported by Eq. \eqref{eq_EE}. Fig. \ref{fig_EE_CRB_lw} illustrates the effect of the communication threshold $\SEth$ on overall EE for $P_{\mathtt{max}}=20$ dBm under the LoR condition. In line with observations from Figs. \ref{fig_EE_CRLB} and \ref{fig_EE_CRLB_Pt2030}, overall EE improves as $\mathtt{CRB}^0$ becomes less stringent. However, increasing $\SEth$ consistently reduces EE for all $\mathtt{CRB}^0$ configurations, primarily due to faster power consumption growth compared to SE gains.

\begin{figure*}[t!]
	\centering
	\subfigure[Overall EE versus $\SEth$.]{%
		\label{fig_EE_R0}%
		\includegraphics[width=0.45\textwidth]{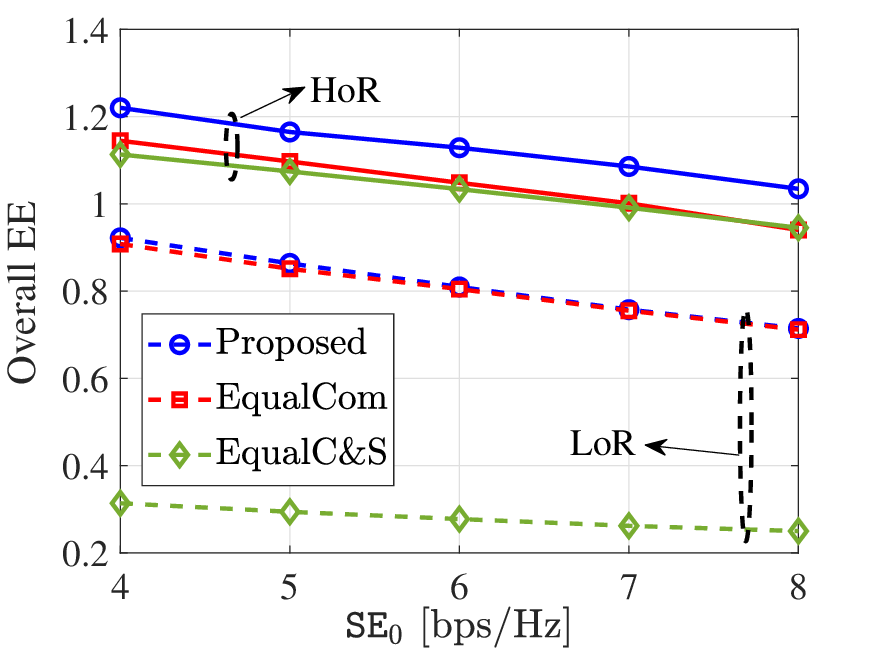}}%
	~
	\subfigure[SE versus $\SEth$.]{%
		\label{fig_SE_R0}%
		\includegraphics[width=0.45\textwidth]{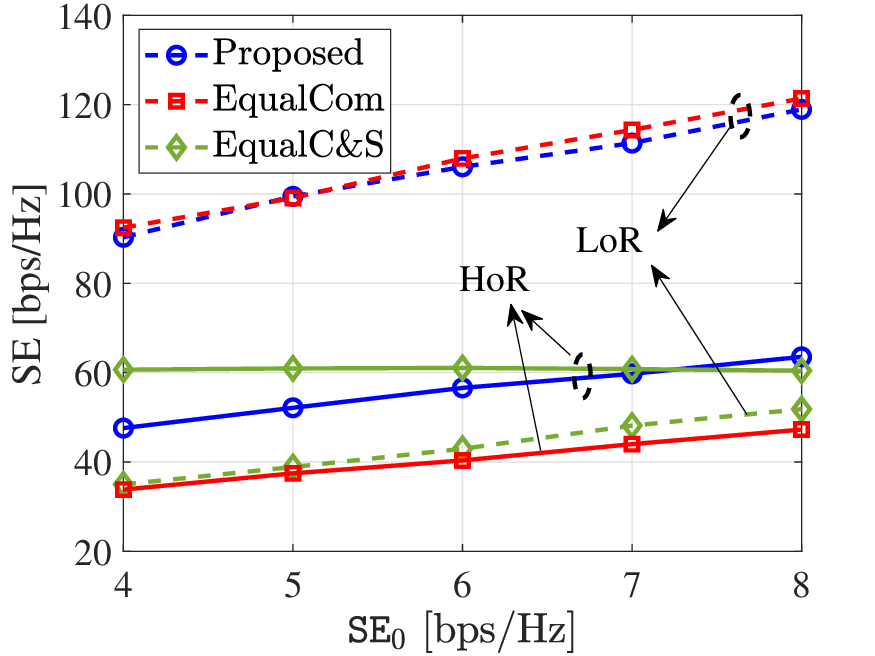}}%
	\\
	\subfigure[Total power consumption versus $\SEth$.]{%
		\label{fig_Ptot_R0}%
		\includegraphics[width=0.45\textwidth]{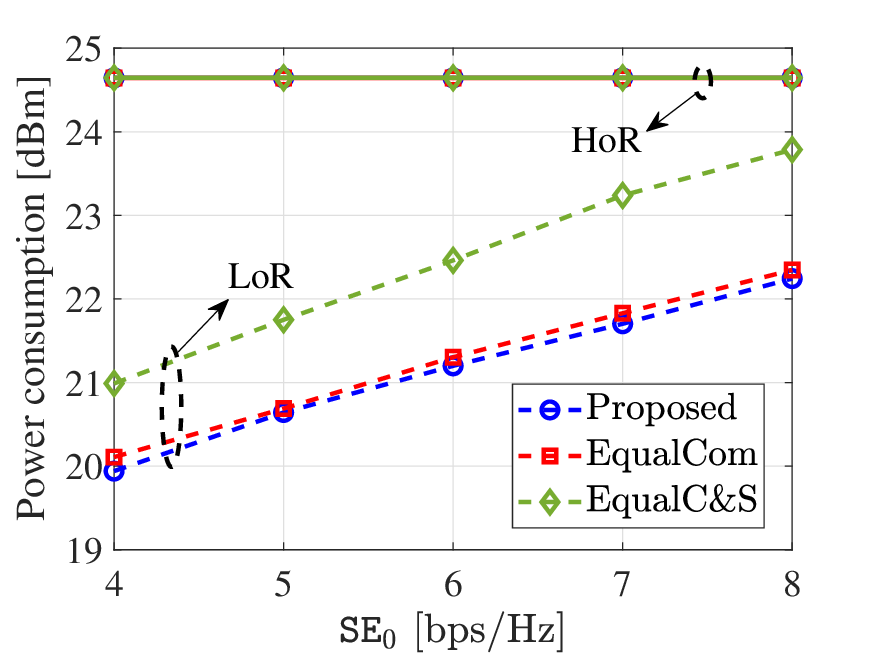}}%
	~
	\subfigure[The converged CRB versus $\SEth$.]{%
		\label{fig_convergedCRLB_R0}%
		\includegraphics[width=0.45\textwidth]{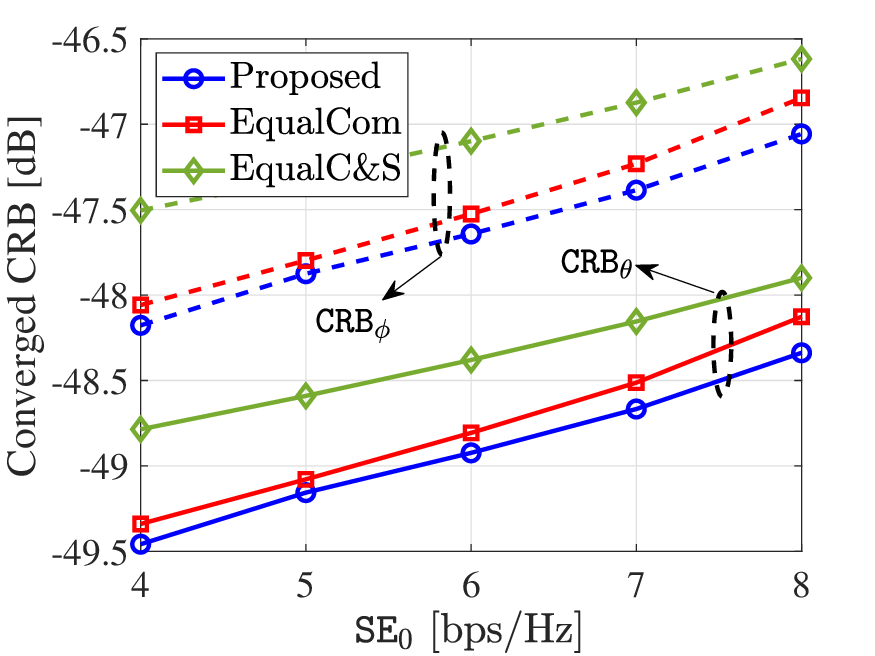}}%
	\caption{The effect of varying $\SEth$ with $\Nt=25$, $\Nr=25$, $Q=16$, $\mathtt{CRB}^0=-35$ dB, and $P_{\mathtt{max}}=20$ dBm.}
	\label{fig_R0}
\end{figure*}

Fig. \ref{fig_R0} evaluates system performance under different thresholds $\SEth$ for $\Nt=25$, $\Nr=25$, $Q=16$, $\mathtt{CRB}^0=-35$ dB, and $P_{\mathtt{max}}=20$ dBm. As illustrated in Fig. \ref{fig_EE_R0}, increasing $\SEth$ consistently degrades the overall EE across all schemes in both the LoR and HoR scenarios. In the LoR, higher communication demands result in increased SE, as shown in Fig. \ref{fig_SE_R0}. However, Fig. \ref{fig_Ptot_R0} reveals that the corresponding rise in power consumption exceeds the SE gains, ultimately reducing system EE.
While the overall EE in the LoR is primarily communication-driven, the HoR is predominantly sensing-dominated. Nevertheless, our simulations reveal a surprising sensitivity of sensing EE to the communication threshold. The joint signal design in \eqref{eq_fkq_bf} imposes a trade-off between communication and sensing power: increasing communication power to meet higher rate thresholds reduces the power available for sensing.
As shown in Fig. \ref{fig_convergedCRLB_R0}, this stringent condition elevates the converged CRBs defined in \eqref{eq_CRB_theta_result} and \eqref{eq_CRB_phi_result}, thereby degrading sensing EE. This decline is clearly evident in Fig. \ref{fig_EE_R0} for the HoR scenario.

\begin{figure}[t]%
	\vspace{-0.5cm}\centering
	\includegraphics[scale=0.55]{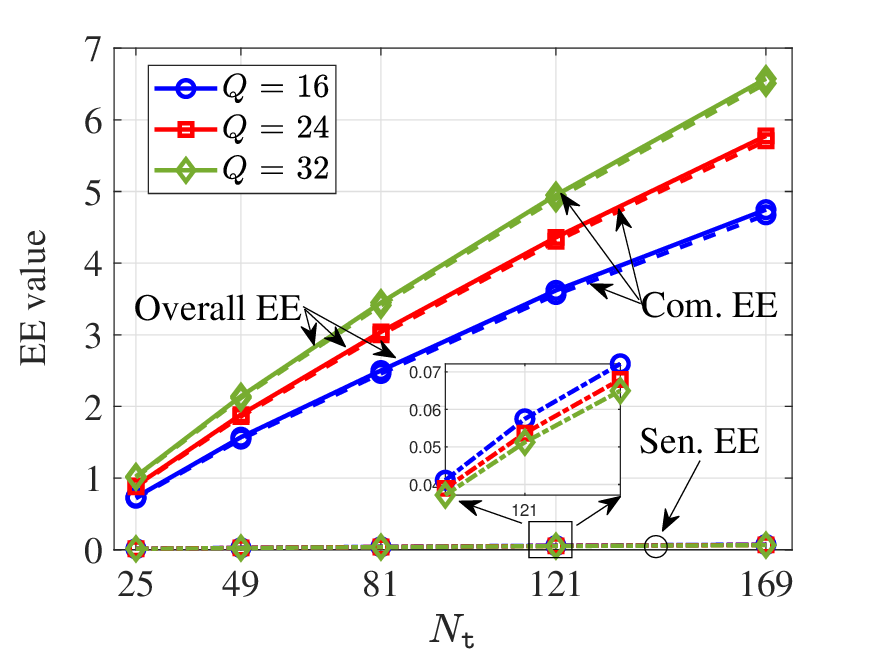}
	\caption[]{EE versus $\Nt$ with $\Nr=25$, $\mathtt{CRB}^0 = -40$ dB, $\SEth = 5.0$ bps/Hz, $P_{\mathtt{max}}=20$ dBm, and $\omega=$1e-4.}%
	\label{fig_Nt}
\end{figure}
\subsection{Effect of $\Nt$ and $Q$}
To evaluate the impact of antenna size $\Nt$ and the number of OFDM channels $Q$, simulations are conducted with $\Nr=25$, $\mathtt{CRB}^0 = -40$ dB, $\SEth = 5.0$ bps/Hz, $P_{\mathtt{max}}=20$ dBm, and $\omega=10^{-4}$. Fig. \ref{fig_Nt} demonstrates that increasing either $\Nt$ or $Q$ enhances system performance. The communication EE (depicted by dashed lines) shows significant improvements, whereas the sensing EE (depicted by dash-dot lines) experiences only marginal gains. The increase in $\Nt$ generally boosts overall EE by leveraging hardware diversity. Additionally, sampling OFDM systems at higher resolutions can substantially enhance EE. However, scaling up $\Nt$ and $Q$ introduces hardware design challenges, particularly in terms of cost and system complexity.

\section{Conclusion}
\label{sec_conclusion}
This work explored the energy efficiency optimization of mono-static mMIMO-OFDM ISAC systems, addressing trade-offs between communication and sensing performance. We formulated a power allocation problem to maximize system EE under communication rate and sensing accuracy (CRB) constraints. To solve this problem, we proposed a solution framework combining problem transformation, CRB analysis, and optimization methods like Dinkelbach's approach and SCA, along with an efficient initialization strategy to improve convergence.
Simulations demonstrated significant EE improvements over baseline methods and revealed that stringent communication thresholds degrade sensing EE, regardless of system priority. These findings highlight the need for adaptive resource allocation strategies in practical ISAC deployments. 
Future work could explore hardware impairments and estimation errors to address the impact of these real-world non-idealities to ISAC systems. In addition, the proposed framework can be extended to incorporate joint CRLB constraints for Doppler and range estimation to support more comprehensive tracking scenarios, possibly requiring joint waveform and power optimization.

\appendices

\renewcommand{\thesectiondis}[2]{\Alph{section}:}
\section{Proof of Lemma \ref{lm_ZF_precoder}} \label{appd_proof_ZF}
\renewcommand{\theequation}{\ref{appd_proof_ZF}.\arabic{equation}}\setcounter{equation}{0}
First, we can express $\mH[q]$ as $\mH[q] = \mZ[q] \mD_{\beta}$, where $\mZ[q] \triangleq [\vz_1[q], \ldots, \vz_K[q]]$, $\vz_k[q] \sim \mathcal{CN}(\bm{0}, \mI_{\Nt})$, and $\mD_{\beta} = \diag{\beta_1, \ldots, \beta_K}$. From \eqref{eq_W_linear}, it follows that
\begin{align*}
	&\mean{\tr{\mW[q] \mW^\H[q]}}\\ 
	&= {\alpha_{\mathtt{ZF}}^2[q]} \mean{\tr{\mD_{\beta}^{-1} (\mZ^\H[q] \mZ[q])^{-1}}} \\
	&= {\alpha_{\mathtt{ZF}}^2[q]} \sum\nolimits_{k=1}^{K} \frac{1}{\beta_k} \mean{\left[ (\mZ^\H[q] \mZ[q])^{-1} \right]_{kk}}\\
	& = {\alpha_{\mathtt{ZF}}^2[q]} \sum\nolimits_{k=1}^{K}  \frac{1}{K\beta_k} \mean{\tr{(\mZ^\H[q] \mZ[q])^{-1}}}\\
	&= {\alpha_{\mathtt{ZF}}^2[q]} \sum_{k=1}^{K} \frac{1}{(\Nt-K)\beta_k} =  \frac{{\alpha_{\mathtt{ZF}}^2[q]}}{\Nt-K} \tr{\mD_{\beta}^{-1}}. \nbthis \label{eq_e_tot_ZF_01}
\end{align*}
Herein we use the property $\meanshort{\trshort{(\mZ^\H[q] \mZ[q])^{-1}}} = \frac{K}{\Nt-K}$ as $\mZ^\H[q] \mZ[q]$ is a central complex Wishart matrix~\cite{yang2013performance}. Setting $\mean{\tr{\mW[q] \mW^\H[q]}} = K, \forall q$, we obtain ${\alpha_{\mathtt{ZF}}[q]}$ as in Lemma \ref{lm_ZF_precoder}.

\renewcommand{\thesectiondis}[2]{\Alph{section}:}
\section{Proof of Lemma \ref{lemma_power}} \label{appd_proof_power}
\renewcommand{\theequation}{\ref{appd_proof_power}.\arabic{equation}}\setcounter{equation}{0}
From $P_{\mathtt{TX}}(\bm{\Omega}) = \mean{\sum_{q=1}^Q \tr{\mX[q] \mX[q]^\H}}$ and using \eqref{eq_mX}, $P_{\mathtt{TX}}(\bm{\Omega})$ can be expanded as:
\begin{align*}
	P_{\mathtt{TX}}(\bm{\Omega}) = \sum\nolimits_{q=1}^Q {\alpha_{\mathtt{ZF}}^2[q]} \varepsilon_{\mathtt{c}}[q] + {\frac{1}{\Nt}} \varepsilon_{\mathtt{s}}[q], \nbthis \label{eq_mean_001}
\end{align*}
where $\varepsilon_{\mathtt{c}}[q] \triangleq \mean{\tr{\mD_{\gamma}[q] {\bm{\mD}_{\xi}^2[q]} \mD_{\gamma}^\H[q] (\mH[q]^\H \mH[q])^{-1}}}$ and $\varepsilon_{\mathtt{s}}[q] \triangleq \mean{\tr{\bar{\bm{\eta}}^\T[q] {\bm{\mD}_{\xi}^2[q]}  \bar{\bm{\eta}}[q] \norm{\va[q]}^2}}$. By \eqref{eq_e_tot_ZF_01}, we obtain
\begin{align*}
	\varepsilon_{\mathtt{c}}[q]
	&\triangleq \sum\nolimits_{k=1}^{K} \frac{\xi_k[q] \gamma_k[q]}{(\Nt-K)\beta_k} =  \frac{\bm{\xi}^\T[q] \mD_{\beta}^{-1} \bm{\gamma}[q]}{\Nt-K}, \nbthis \label{eq_eps_c} \\
	\varepsilon_{\mathtt{s}}[q]
	&=\triangleq \Nt \sum\nolimits_{k=1}^K \xi_k[q] \eta_k[q] = \Nt \bm{\xi}^\T[q] \bm{\eta}[q], \nbthis \label{eq_eps_s}
\end{align*}
where $\bm{\xi}[q] \triangleq [\xi_1[q], \ldots, \xi_K[q]]^\T$, $\bm{\gamma}[q] \triangleq [\gamma_1[q], \ldots, \gamma_K[q]]^\T$, and $\bm{\eta}[q] \triangleq [\eta_1[q], \ldots, \eta_K[q]]^\T$. We omit the detailed derivations for \eqref{eq_eps_c} and \eqref{eq_eps_s} due to space limitations. 
Finally, \eqref{eq_Ptx} follows \eqref{eq_mean_001}--\eqref{eq_eps_s}.

\renewcommand{\thesectiondis}[2]{\Alph{section}:}
\section{Proof of Theorem \ref{theo_CRB}} \label{appd_proof_CRB}
\renewcommand{\theequation}{\ref{appd_proof_CRB}.\arabic{equation}}\setcounter{equation}{0}
The equivalent Fisher Information Matrix (FIM) for $(\theta, \phi)$ over all subcarriers is given as \cite{wu2011antenna} 
\begin{align*}
	\!\mT_{\mathtt{FIM}} \!&=  \!\sum_{q=1}^Q
	\begin{bmatrix}
		\Ttt[q]  & \Ttp[q] \\
		\Ttp[q] \!& \!\Tpp[q] \!-\! \Tpa[q] \Taa^{-1}[q] \Tpa^\T[q]
	\end{bmatrix}. \nbthis \label{eq_FIM}
\end{align*}
In \eqref{eq_FIM}, the elements of $\mT_{\mathtt{FIM}}[q]$ are given by \cite{bekkerman2006target}
\begin{align*}
	\Ttt[q] &= \kappa \abs{\alpha}^2 \tr{\dot{\mG}_{\theta}[q] \mR_x[q] \dot{\mG}^\H_{\theta}[q]}, \nbthis \label{eq_def_Jtt} \\
	\Ttp[q] &= \kappa \abs{\alpha}^2 \tr{\dot{\mG}_{\phi}[q] \mR_x[q] \dot{\mG}^\H_{\theta}[q]}, \nbthis \label{eq_def_Jtp}\\
	\Tta[q] &= \kappa \re{\alpha^* \tr{\mG[q] \mR_x[q] \dot{\mG}^\H_{\theta}[q]}[1,j]}, \nbthis \label{eq_def_Jta}\\
	\Tpp[q] &= \kappa \abs{\alpha}^2 \tr{\dot{\mG}_{\phi}[q] \mR_x[q] \dot{\mG}^\H_{\phi}[q]}, \nbthis \label{eq_def_Jpp} \\
	\Tpa[q] &= \kappa \re{\alpha^* \tr{\mG[q] \mR_x[q] \dot{\mG}^\H_{\phi}[q]}[1,j]}, \nbthis \label{eq_def_Jpa}\\
	\Taa[q] &= \kappa \tr{\mG[q] \mR_x[q] \mG^\H[q]} \mI_2, \nbthis \label{eq_def_Jaa}
\end{align*}
where $\kappa \triangleq \frac{2L}{\sigmas}$, and $\mR_x[q] = \frac{1}{L} \mean{\mX[q] \mX[q]^\H}$. After some algebraic manipulation, we obtain
\begin{align*}
	\mR_x[q] &= \frac{1}{\Nt} \left( \bm{\xi}^\T[q] \bar{\mD}_{\beta} \bm{\gamma}[q] \mI_{\Nt} + \va[q] \bm{\xi}^\T[q] \bm{\eta}[q] \va^\H[q] \right). \nbthis \label{eq_approx_cov}
\end{align*}
Substituting \eqref{eq_approx_cov} into \eqref{eq_def_Jtt}--\eqref{eq_def_Jaa}, we obtain \eqref{eq_result_Ttt}--\eqref{eq_result_Taa}. As a result, the CRB associated with \(\theta\) and \(\phi\) are derived as $\mathtt{CRB}_{\theta}(\bm{\Omega}) = \left[\mT_{\mathtt{FIM}}^{-1}\right]_{11}$ and $\mathtt{CRB}_{\phi}(\bm{\Omega}) = \left[\mT_{\mathtt{FIM}}^{-1}\right]_{22}$, where $[\mA]_{ij}$ denotes the $(i,j)$-th entry of matrix $\mA$. Using \eqref{eq_FIM}--\eqref{eq_approx_cov}, we obtain the results in Theorem \ref{theo_CRB}.

	\renewcommand{\thesectiondis}[2]{\Alph{section}:}
		\section{Preliminary Inequality} \label{Preliminary_Inequality}
		\renewcommand{\theequation}{\ref{Preliminary_Inequality}.\arabic{equation}}\setcounter{equation}{0}
			Consider the following products
			\begin{align*}
				f_{p^2}(x,y) \triangleq xy, \text{and } f_{p^3}(t,x,y) \triangleq t\sqrt{xy}, (t,x,y) \in \setR_{++}^3.
			\end{align*}
			Global bounds for $f_{p^2}(x,y)$ are found by
			\begin{align*}
				f_{p^2}(x,y) &\!\geq\! \Ldp{x}{y} \!\triangleq\! (\bar{x}\!+\!\bar{y})(x\!+\!y) \!-\! 0.5 \!\Big[(\bar{x}\!+\!\bar{y})^2 \!+\!  x^2 \!+\! y^2\Big], \nbthis \label{inequality_xy_geq} \\
				f_{p^2}(x,y) &\!\leq\! \Udp{x}{y}  \!\triangleq\! \frac{0.5\bar{y}}{\bar{x}} x^2 + \frac{0.5\bar{x}}{\bar{y}} y^2, \nbthis \label{inequality_xy_leq}
			\end{align*}
			and those for $f_{p^3}(t,x,y)$ by
			\begin{align*}
				f_{p^3}(t,x,y) 
				&\!\geq\! t \Ldp{\sqrt{x}}{\sqrt{y}} \!\geq\! \Ltp{t}{\sqrt{x}}{\sqrt{y}}\\
				&\!\triangleq\! (\sqrt{\bar{x}}\!+\!\sqrt{\bar{y}})(\Ldp{t}{\sqrt{x}}\!+\!\Ldp{t}{\sqrt{y}}) \\
				& \!-\! 0.5 \!\Big[(\sqrt{\bar{x}}\!+\!\sqrt{\bar{y}})^2 t \!+\!  \Udp{t}{x} \!+\! \Udp{t}{y}\Big], \nbthis \label{inequality_txy_geq} \\
				f_{p^3}(t,x,y) 
				&\!\leq\! t \Udp{\sqrt{x}}{\sqrt{y}} \!\leq\! \Utp{t}{\sqrt{x}}{\sqrt{y}}\\
				&\!\triangleq\! 0.5\sqrt{\bar{y}/\bar{x}} \Udp{t}{x} + 0.5 \sqrt{\bar{x}/\bar{y}} \Udp{t}{y}, \nbthis \label{inequality_txy_leq}
			\end{align*}
			where $\Ldp{x}{y}$ and $\Ldp{\sqrt{x}}{\sqrt{y}}$ are concave lower bounds of $f_{p^2}$ and $f_{p^3}(t,x,y)$, respectively, and $\Udp{x}{y}$ and $\Udp{\sqrt{x}}{\sqrt{y}}$ are their corresponding convex upper bounds. These preliminaries are applicable for any sets of feasible points that satisfy $(\bar{t},\bar{x},\bar{y}) \in \setR_{++}^3$.
		
\begingroup
\bibliographystyle{IEEEtran}
\bibliography{IEEEabrv,Bibliography}

\begin{thebibliography}{10}
\providecommand{\url}[1]{#1}
\csname url@samestyle\endcsname
\providecommand{\newblock}{\relax}
\providecommand{\bibinfo}[2]{#2}
\providecommand{\BIBentrySTDinterwordspacing}{\spaceskip=0pt\relax}
\providecommand{\BIBentryALTinterwordstretchfactor}{4}
\providecommand{\BIBentryALTinterwordspacing}{\spaceskip=\fontdimen2\font plus
\BIBentryALTinterwordstretchfactor\fontdimen3\font minus \fontdimen4\font\relax}
\providecommand{\BIBforeignlanguage}[2]{{%
\expandafter\ifx\csname l@#1\endcsname\relax
\typeout{** WARNING: IEEEtran.bst: No hyphenation pattern has been}%
\typeout{** loaded for the language `#1'. Using the pattern for}%
\typeout{** the default language instead.}%
\else
\language=\csname l@#1\endcsname
\fi
#2}}
\providecommand{\BIBdecl}{\relax}
\BIBdecl

\bibitem{chowdhury20206g}
M.~Z. Chowdhury, M.~Shahjalal, S.~Ahmed, and Y.~M. Jang, ``{6G} wireless communication systems: Applications, requirements, technologies, challenges, and research directions,'' \emph{{IEEE} IEEE Open J. Commun. Society}, vol.~1, pp. 957--975, 2020.

\bibitem{fang2022joint}
X.~Fang, W.~Feng, Y.~Chen, N.~Ge, and Y.~Zhang, ``Joint communication and sensing: {M}odels and potentials of using {MIMO},'' \emph{arXiv preprint arXiv:2205.09409}, 2022.

\bibitem{liu2022integrated}
F.~Liu, Y.~Cui, C.~Masouros, J.~Xu, T.~X. Han, Y.~C. Eldar, and S.~Buzzi, ``Integrated sensing and communications: {T}owards dual-functional wireless networks for {6G} and beyond,'' \emph{{IEEE} J. Sel. Areas Commun.}, 2022.

\bibitem{liu2022survey}
A.~Liu, Z.~Huang, M.~Li, Y.~Wan, W.~Li, T.~X. Han, C.~Liu, R.~Du, D.~K.~P. Tan, J.~Lu \emph{et~al.}, ``A survey on fundamental limits of integrated sensing and communication,'' \emph{IEEE Commun. Surveys Tuts.}, vol.~24, no.~2, pp. 994--1034, 2022.

\bibitem{zhang2021overview}
J.~A. Zhang, F.~Liu, C.~Masouros, R.~W. Heath, Z.~Feng, L.~Zheng, and A.~Petropulu, ``An overview of signal processing techniques for joint communication and radar sensing,'' \emph{{IEEE} J. Sel. Topics Signal Process.}, vol.~15, no.~6, pp. 1295--1315, 2021.

\bibitem{nguyen2024performance}
N.~T. Nguyen, V.-D. Nguyen, H.~V. Nguyen, H.~Q. Ngo, A.~L. Swindlehurst, and M.~Juntti, ``Performance analysis and power allocation for massive mimo isac systems,'' \emph{{IEEE} Trans. Signal Process.}, vol.~73, pp. 1691--1707, 2025.

\bibitem{zhang2023isac}
T.~Zhang, G.~Li, S.~Wang, G.~Zhu, G.~Chen, and R.~Wang, ``{ISAC-accelerated edge intelligence: Framework, optimization, and analysis},'' \emph{{IEEE} Trans. Green Commun. Network.}, vol.~7, no.~1, pp. 455--468, 2023.

\bibitem{nguyen2023multiuser}
N.~T. Nguyen, N.~Shlezinger, Y.~C. Eldar, and M.~Juntti, ``Multiuser {MIMO} wideband joint communications and sensing system with subcarrier allocation,'' \emph{{IEEE} Trans. Signal Process.}, vol.~71, pp. 2997--3013, 2023.

\bibitem{nguyen2024joint}
N.~T. Nguyen, L.~V. Nguyen, N.~Shlezinger, Y.~C. Eldar, A.~L. Swindlehurst, and M.~Juntti, ``Joint communications and sensing hybrid beamforming design via deep unfolding,'' \emph{{IEEE} J. Sel. Topics Signal Process.}, vol.~18, no.~5, pp. 901--916, 2024.

\bibitem{nguyen2021hierarchical}
H.~T. Nguyen, D.~T. Hoang, N.~C. Luong, D.~Niyato, and D.~I. Kim, ``A hierarchical game model for ofdm integrated radar and communication systems,'' \emph{EEE Trans. Veh. Technol.}, vol.~70, no.~5, pp. 5077--5082, 2021.

\bibitem{liu2020joint}
X.~Liu, T.~Huang, N.~Shlezinger, Y.~Liu, J.~Zhou, and Y.~C. Eldar, ``Joint transmit beamforming for multiuser {MIMO} communications and {MIMO} radar,'' \emph{{IEEE} Trans. Signal Process.}, vol.~68, pp. 3929--3944, 2020.

\bibitem{johnston2022mimo}
J.~Johnston, L.~Venturino, E.~Grossi, M.~Lops, and X.~Wang, ``{MIMO OFDM dual-function radar-communication under error rate and beampattern constraints},'' \emph{{IEEE} J. Sel. Areas Commun.}, vol.~40, no.~6, pp. 1951--1964, 2022.

\bibitem{choi2024joint}
J.~Choi, J.~Park, N.~Lee, and A.~Alkhateeb, ``Joint and robust beamforming framework for integrated sensing and communication systems,'' \emph{arXiv preprint arXiv:2402.09155}, 2024.

\bibitem{wang2023optimizing}
Y.~Wang, Z.~Yang, J.~Cui, P.~Xu, G.~Chen, T.~Q. Quek, and R.~Tafazolli, ``Optimizing the fairness of {STAR-RIS} and {NOMA} assisted integrated sensing and communication systems,'' \emph{{IEEE} Trans. Wireless Commun.}, 2023.

\bibitem{song2023intelligent}
X.~Song, J.~Xu, F.~Liu, T.~X. Han, and Y.~C. Eldar, ``Intelligent reflecting surface enabled sensing: {Cram{\'e}r-Rao} bound optimization,'' \emph{{IEEE} Trans. Signal Process.}, vol.~71, pp. 2011--2026, 2023.

\bibitem{ren2022fundamental}
Z.~Ren, X.~Song, Y.~Fang, L.~Qiu, and J.~Xu, ``Fundamental {CRB}-rate tradeoff in multi-antenna multicast channel with {ISAC},'' in \emph{Proc. IEEE Global Commun. Conf.}, 2022, pp. 1261--1266.

\bibitem{liu2021cramer}
F.~Liu, Y.-F. Liu, A.~Li, C.~Masouros, and Y.~C. Eldar, ``{Cram{\'e}r-Rao bound optimization for joint radar-communication beamforming},'' \emph{{IEEE} Trans. Signal Process.}, vol.~70, pp. 240--253, 2021.

\bibitem{nguyen2020nonsmooth}
H.~T. Nguyen, H.~D. Tuan, T.~Q. Duong, H.~V. Poor, and W.-J. Hwang, ``Nonsmooth optimization algorithms for multicast beamforming in content-centric fog radio access networks,'' \emph{EEE Trans. Signal Process.}, vol.~68, pp. 1455--1469, 2020.

\bibitem{nguyen2017spectral}
V.-D. Nguyen, T.~Q. Duong, H.~D. Tuan, O.-S. Shin, and H.~V. Poor, ``Spectral and energy efficiencies in full-duplex wireless information and power transfer,'' \emph{IEEE Trans. Commun.}, vol.~65, no.~5, pp. 2220--2233, 2017.

\bibitem{leong2024green}
W.~Y. Leong, Y.~Z. Leong, and W.~San~Leong, ``Green communication systems: Towards sustainable networking,'' in \emph{5th Int. Conf. Inf. Science, Parallel Distributed Systems}, 2024, pp. 559--564.

\bibitem{gandotra2017green}
P.~Gandotra, R.~K. Jha, and S.~Jain, ``Green communication in next generation cellular networks: A survey,'' \emph{IEEE Access}, vol.~5, pp. 11\,727--11\,758, 2017.

\bibitem{zou2024energy}
J.~Zou, S.~Sun, C.~Masouros, Y.~Cui, Y.-F. Liu, and D.~W.~K. Ng, ``Energy-efficient beamforming design for integrated sensing and communications systems,'' \emph{{IEEE} Trans. Commun.}, vol.~72, no.~6, pp. 3766--3782, 2024.

\bibitem{allu2024robust}
R.~Allu, M.~Katwe, K.~Singh, T.~Q. Duong, and C.-P. Li, ``Robust energy efficient beamforming design for isac full-duplex communication systems,'' \emph{IEEE Wireless Commun. Lett.}, 2024.

\bibitem{wu2023energy}
G.~Wu, Y.~Fang, J.~Xu, Z.~Feng, and S.~Cui, ``Energy-efficient mimo integrated sensing and communications with on-off non-transmission power,'' \emph{IEEE Internet Things J.}, 2023.

\bibitem{he2022energy}
Z.~He, W.~Xu, H.~Shen, Y.~Huang, and H.~Xiao, ``Energy efficient beamforming optimization for integrated sensing and communication,'' \emph{{IEEE} Wireless Commun. Lett.}, vol.~11, no.~7, pp. 1374--1378, 2022.

\bibitem{hatami2025energy}
M.~Hatami, N.~T. Nguyen, and M.~Juntti, ``Energy efficient waveform design and subcarrier allocation for multicarrier {MIMO JCAS},'' in \emph{Proc. IEEE Wireless Commun. and Networking Conf.}, 2025.

\bibitem{liu2018mu}
F.~Liu, C.~Masouros, A.~Li, H.~Sun, and L.~Hanzo, ``{MU-MIMO communications with MIMO radar: From co-existence to joint transmission},'' \emph{{IEEE} Trans. Wireless Commun.}, vol.~17, no.~4, pp. 2755--2770, 2018.

\bibitem{cheng2021hybrid}
Z.~Cheng, Z.~He, and B.~Liao, ``Hybrid beamforming design for ofdm dual-function radar-communication system,'' \emph{IEEE J. Sel. Topics Signal Process.}, vol.~15, no.~6, pp. 1455--1467, 2021.

\bibitem{nguyen2023jointssp}
N.~T. Nguyen, N.~Shlezinger, K.-H. Ngo, V.-D. Nguyen, and M.~Juntti, ``Joint communications and sensing design for multi-carrier {MIMO} systems,'' \emph{Proc. IEEE Works. on Statistical Signal Processing}, 2023.

\bibitem{mylonopoulos2024extended}
G.~Mylonopoulos, B.~Makki, G.~Fodor, and S.~Buzzi, ``Extended {ARQ} protocol for reliable integrated sensing and communication systems,'' in \emph{Proc. IEEE Works. on Sign. Proc. Adv. in Wirel. Comms.}, 2024, pp. 321--325.

\bibitem{huang2020majorcom}
T.~Huang, N.~Shlezinger, X.~Xu, Y.~Liu, and Y.~C. Eldar, ``{MAJoRCom}: A dual-function radar communication system using index modulation,'' \emph{{IEEE} Trans. Signal Process.}, vol.~68, pp. 3423--3438, 2020.

\bibitem{ma2021spatial}
D.~Ma, N.~Shlezinger, T.~Huang, Y.~Shavit, M.~Namer, Y.~Liu, and Y.~C. Eldar, ``Spatial modulation for joint radar-communications systems: Design, analysis, and hardware prototype,'' \emph{{IEEE} Trans. Veh. Technol.}, vol.~70, no.~3, pp. 2283--2298, 2021.

\bibitem{ma2021frac}
D.~Ma, N.~Shlezinger, T.~Huang, Y.~Liu, and Y.~C. Eldar, ``{FRaC}: {FMCW}-based joint radar-communications system via index modulation,'' \emph{{IEEE} J. Sel. Topics Signal Process.}, vol.~15, no.~6, pp. 1348--1364, 2021.

\bibitem{liu2017robust}
F.~Liu, C.~Masouros, A.~Li, and T.~Ratnarajah, ``{Robust MIMO beamforming for cellular and radar coexistence},'' \emph{{IEEE} Trans. Wireless Commun.}, vol.~6, no.~3, pp. 374--377, 2017.

\bibitem{liu2022transmit}
X.~Liu, T.~Huang, Y.~Liu, and Y.~C. Eldar, ``Transmit beamforming with fixed covariance for integrated {MIMO} radar and multiuser communications,'' in \emph{Proc. IEEE Int. Conf. Acoust., Speech, Signal Processing}, 2022, pp. 8732--8736.

\bibitem{pritzker2022transmit}
J.~Pritzker, J.~Ward, and Y.~C. Eldar, ``Transmit precoder design approaches for dual-function radar-communication systems,'' \emph{arXiv preprint arXiv:2203.09571}, 2022.

\bibitem{pritzker2021transmit}
------, ``Transmit precoding for dual-function radar-communication systems,'' in \emph{Proc. Annual Asilomar Conf. Signals, Syst., Comp.}, 2021.

\bibitem{elbir2021terahertz}
A.~M. Elbir, K.~V. Mishra, and S.~Chatzinotas, ``Terahertz-band joint ultra-massive {MIMO} radar-communications: Model-based and model-free hybrid beamforming,'' \emph{{IEEE} J. Sel. Topics Signal Process.}, vol.~15, no.~6, pp. 1468--1483, 2021.

\bibitem{mollen2016uplink}
C.~Mollen, J.~Choi, E.~G. Larsson, and R.~W. Heath, ``{Uplink performance of wideband massive MIMO with one-bit ADCs},'' \emph{{IEEE} Trans. Wireless Commun.}, vol.~16, no.~1, pp. 87--100, 2016.

\bibitem{yu2016alternating}
X.~Yu, J.-C. Shen, J.~Zhang, and K.~B. Letaief, ``Alternating minimization algorithms for hybrid precoding in millimeter wave {MIMO} systems,'' \emph{{IEEE} J. Sel. Topics Signal Process.}, vol.~10, no.~3, pp. 485--500, 2016.

\bibitem{Beck:JGO:10}
A.~Beck, A.~Ben-Tal, and L.~Tetruashvili, ``A sequential parametric convex approximation method with applications to nonconvex truss topology design problems,'' \emph{J. Global Optim.}, vol.~47, no.~1, pp. 29--51, May 2010.

\bibitem{NasirTWC21}
A.~A. Nasir, H.~D. Tuan, H.~H. Nguyen, M.~Debbah, and H.~V. Poor, ``Resource allocation and beamforming design in the short blocklength regime for {URLLC},'' \emph{IEEE Trans. Wireless Commun.}, vol.~20, no.~2, pp. 1321--1335, 2021.

\bibitem{sheng2018power}
Z.~Sheng, H.~D. Tuan, A.~A. Nasir, T.~Q. Duong, and H.~V. Poor, ``Power allocation for energy efficiency and secrecy of wireless interference networks,'' \emph{IEEE Trans. Wireless Commun.}, vol.~17, no.~6, pp. 3737--3751, 2018.

\bibitem{peaucelle2002user}
\BIBentryALTinterwordspacing
D.~Peaucelle, D.~Henrion, Y.~Labit, and K.~Taitz, ``User’s guide for {SeDuMi} interface 1.04,'' 2002. [Online]. Available: \url{http://homepages.laas.fr/peaucell/software/sdmguide.pdf}
\BIBentrySTDinterwordspacing

\bibitem{kha2011fast}
H.~H. Kha, H.~D. Tuan, and H.~H. Nguyen, ``Fast global optimal power allocation in wireless networks by local dc programming,'' \emph{IEEE Trans. Wireless Commun.}, vol.~11, no.~2, pp. 510--515, 2011.

\bibitem{ngo2013energy}
H.~Q. Ngo, E.~G. Larsson, and T.~L. Marzetta, ``Energy and spectral efficiency of very large multiuser {MIMO} systems,'' \emph{{IEEE} Trans. Commun.}, vol.~61, no.~4, pp. 1436--1449, 2013.

\bibitem{xiong2011energy}
C.~Xiong, G.~Y. Li, S.~Zhang, Y.~Chen, and S.~Xu, ``Energy-and spectral-efficiency tradeoff in downlink ofdma networks,'' \emph{IEEE Trans. Wireless Commun.}, vol.~10, no.~11, pp. 3874--3886, 2011.

\bibitem{yang2013performance}
H.~Yang and T.~L. Marzetta, ``Performance of conjugate and zero-forcing beamforming in large-scale antenna systems,'' \emph{{IEEE} J. Sel. Areas Commun.}, vol.~31, no.~2, pp. 172--179, 2013.

\bibitem{wu2011antenna}
X.~H. Wu, A.~A. Kishk, and A.~W. Glisson, ``Antenna effects on a monostatic {MIMO} radar for direction estimation, a {C}ram{\`e}r-{R}ao lower bound analysis,'' \emph{{IEEE} Trans. Antennas Propag.}, vol.~59, no.~6, pp. 2388--2395, 2011.

\bibitem{bekkerman2006target}
I.~Bekkerman and J.~Tabrikian, ``Target detection and localization using {MIMO} radars and sonars,'' \emph{{IEEE} Trans. Signal Process.}, vol.~54, no.~10, pp. 3873--3883, 2006.

\end{thebibliography}
\endgroup

\end{document}